\documentclass[journal,onecolumn,11pt]{IEEEtran}
\usepackage{subfigure,cite,graphicx,amsmath,amssymb,eufrak,mathrsfs,epsf,epsfig}
\usepackage{multicol,lipsum}
\usepackage{xfrac}
\usepackage{diagbox}
\usepackage{algorithm}
\usepackage{algpseudocode}
\usepackage{makecell}
\usepackage{verbatim}

\ifCLASSINFOpdf

\else
\fi

\begin{document}
\newtheorem{ach}{Achievability}
\newtheorem{con}{Converse}
\newtheorem{definition}{Definition}
\newtheorem{theorem}{Theorem}
\newtheorem{lemma}{Lemma}
\newtheorem{example}{Example}
\newtheorem{cor}{Corollary}
\newtheorem{prop}{Proposition}
\newtheorem{conjecture}{Conjecture}
\newtheorem{remark}{Remark}
\title{Reconfigurable Intelligent Surface Aided Hybrid Beamforming: Optimal Placement and Beamforming Design}\author{\IEEEauthorblockN{Najam Us Saqib,~\IEEEmembership{Graduate Student Member,~IEEE}, Shumei Hou,\\ Sung Ho Chae,~\IEEEmembership{Member,~IEEE}, and Sang-Woon Jeon,~\IEEEmembership{Member,~IEEE}
\thanks{ This paper will be presented in part at the IEEE International Conference on Communications (ICC), May 2023 \cite{myiccpaper}.}
\thanks{N. U. Saqib and S. Hou are with the Department of Electrical and Electronic Engineering, Hanyang University, Ansan 15588, South Korea (e-mail: \{najam, h2021156963\}@hanyang.ac.kr).}
\thanks{S. H. Chae is with the Department of Electronic Engineering, Kwangwoon University, Seoul 01897, South Korea (e-mail: sho.chae00@gmail.com).}
\thanks{S.-W. Jeon, \emph{corresponding author}, is with the Department of Electrical and Electronic Engineering, Hanyang University, Ansan 15588, South Korea (e-mail: sangwoonjeon@hanyang.ac.kr).}
}}
 \maketitle
\vspace{-0.4in}
\begin{abstract}
We consider reconfigurable intelligent surface (RIS) aided sixth-generation (6G) terahertz (THz) communications for indoor environment in which a base station (BS) wishes to send independent messages to its serving users with the help of multiple RISs. For indoor environment, various obstacles such as pillars, walls, and other objects can result in no line-of-sight signal path between the BS and a user, which can significantly degrade performance. To overcome such limitation of indoor THz communication, we firstly optimize the placement of RISs to maximize the coverage area. Under the optimized RIS placement, we propose 3D hybrid beamforming at the BS and phase adjustment at RISs, which are jointly performed at the BS and RISs via codebook-based 3D beam scanning with low complexity. Numerical simulations demonstrate that the proposed scheme significantly improves the average sum rate compared to the cases of no RIS and randomly deployed RISs. It is further shown that the proposed codebook-based 3D beam scanning efficiently aligns analog beams between BS--user links or BS--RIS--user links and, as a consequence, achieves the average sum rate close to that of coherent beam alignment requiring global channel state information.
\end{abstract}

\begin{IEEEkeywords}
Hybrid beamforming, machine learning, reconfigurable intelligent surface (RIS), RIS placement optimization, terahertz (THz) communication. 
\end{IEEEkeywords}
 \IEEEpeerreviewmaketitle

\section{Introduction}
For the future sixth-generation (6G) communication systems, extremely high data rate services such as holographic displays, auto driving,  interactive online gaming, and tactile and haptic internet applications are expected to begin in earnest~\cite{Dang:20,Giordani:20,TATARIA:21,Chae-Jeon:21}. 
In order to support such soaring data traffic required from new and emerging applications, 6G communication systems are required to support data rates more than $1$ Tbps~\cite{ref1}, which is at least $100$ times higher than that of the fifth-generation (5G) communication.
For instance, holograms typically require data rates of tens of Mbps to a few Tbps, and the data rates that access points in metro stations or shopping malls need to support could be up to $1$ Tbps~\cite{TATARIA:21}.
To accommodate these data demands in 6G systems, terahertz (THz) communication technologies have been actively developed due to the advantage of being able to exploit a wide frequency band~\cite{Koenig:13,Chen:19,Elayan:20}.

Even though the utilization of broadband is essential to provide such improved data rates, there are some challenging problems to be solved in order to integrate ultra-high frequency communications into 6G systems.
For ultra-high frequency communications using millimeter wave (mmWave) or THz  frequency bands, it has been reported that the received signal can be extremely weak due to significant path loss~\cite{Rappaport:17,Hemadeh18} compared to the conventional communication bands, which might severely restrict the coverage of 6G systems.     
Moreover, at ultra-high frequencies such as mmWave and THz, diffuse reflection becomes significant, and as a consequence, the effective number of multi-path components generally decreases as the operating frequency increases.
Therefore, communication channels become sensitive to network environment and also can be rank-deficient. More importantly, due to these poor-scattering characteristics of ultra-high frequency bands, it will be very difficult to secure stable communication channels when line-of-sight (LoS) channel components are blocked by buildings, walls, or other obstacles~\cite{ref2,ref3,ref4,ref5}.     

In order to overcome such limitation, hybrid beamforming techniques that can improve both cell coverage and spectral efficiency of current 5G systems have been actively studied~\cite{Rappaport:17,Hemadeh18,Giordani:19,Nadeem19, Jeon1, Jeon2, Jeon3}. In hybrid antenna array structures, a subset of antenna elements (or sub-array) is connected to an antenna port or a radio frequency (RF) chain so that the implementation cost and complexity can be efficiently reduced.
As classified in~\cite{Molisch:17,Song:20}, hybrid antenna structures can be categorized into fully-connected and sub-array structures. For the first case, each antenna port is connected to all antenna elements, whereas each antenna port is connected to a disjoint sub-array for the second case. 
To reduce the computational complexity of hybrid beamforming designs, a two-step approach has been widely adopted, in which the analog beamforming part usually performs coherent beamforming to increase the received signal-to-noise ratio (SNR) and the digital processing part focuses on eliminating interference between multiple data streams~\cite{Nadeem19,Zhang:20,Ayach:14,Liang:14, Alkhateeb:15, Ni:16,Rajashekar:17, Hu:18}. Several practical low-complexity analog and digital beamforming strategies have been summarized in~\cite{Nadeem19, Zhang:20}. In particular, general downtilting analog beamforming methods for both single-cell and multi-cell systems have been stated in~\cite{Nadeem19}  and performance comparisons for several beamforming algorithms have been demonstrated in~\cite{Zhang:20}.

Although hybrid beamforming techniques have been successfully adopted to 5G systems and provided an extended coverage for 5G mmWave communications~\cite{Rappaport:17,Hemadeh18,Giordani:19,Nadeem19,Zhang:20}, a straightforward extension of such techniques to 6G systems utilizing much higher frequency bands is expected to have several limitations. Firstly, as the operating frequency band increases, extremely massive multiple-input and multiple-output (MIMO) with much more antenna elements will be required to offer cell coverage for 6G THz communications, but it is expected at the same time that the system complexity and cost increase significantly, making it impossible to implement in practical systems.
Secondly, hybrid beamforming with the extremely massive MIMO setup cannot fundamentally resolve the signal blockage phenomenon of poor-scattering or LoS-only channels. A promising approach to address these technical challenges to provide improved coverage and throughput for ultra-high frequency communications is the introduction of reconfigurable intelligent surfaces (RISs) into 6G systems.

Unlike the conventional relays or reflectors, RIS is a planar surface consisting of a large number of reflecting components, each of which is capable of adjusting the angle and/or amplitude of its reflected signal in a controllable and intelligent manner~\cite{ref6, ref7, ref9, ref10}. 
Thanks to such controllable degrees of freedom, RIS can be utilized for a variety of purposes, including to address signal blockage, the key hole effect, and multi-path loss or penetration loss, which are essentially required to secure stable channels for ultra-high frequency communications. There are several ways to implement RIS using passive reflectarrays or metasurfaces connected with electronic devices such as positive-intrinsic-negative (PIN) diodes, field-effect transistors (FETs), or micro-electro-mechanical system (MEMS) switches~\cite{ref7, ref11}.

Because of various potentials of RIS described above, it has been considered as one of the fundamental emerging components for 6G communication networks. In order to fully utilize RISs and successfully integrate them into 6G beamforming systems, the joint design of multi-antenna beamforming at base stations (BSs) and/or users and RIS reflecting procedure has been actively studied in the literature. In~\cite{ref12,ref13,tao:21}, a RIS-aided multiple-input and single-output (MISO) system in which a single multi-antenna BS serves multiple single-antenna users with the help of a RIS has been considered. In particular, the total transmit power minimization problem has been solved subject to the signal-to-interference-plus-noise ratio (SINR) constraint for each user in~\cite{ref12}, and the maximum SINR achieved by the worst-case user based on the linear precoder under the RIS power constraint has been solved in~\cite{ref13}. A deep neural network to parameterize the mapping from the received pilots to optimized RIS reflection coefficients has been proposed in~\cite{tao:21}.
The RIS-aided beamforming techniques have also been extended to MIMO and massive MIMO systems~\cite{ref14,ref15,9906810,ref16,ref18}. Specifically, the capacity characterization of RIS-aided MIMO systems has been studied in~\cite{ref14} by jointly optimizing the RIS reflection coefficients and the MIMO transmit covariance matrix, and an iterative algorithm that based on the projected gradient method has been proposed in~\cite{ref15}. 
A general MIMO setup consisting of hybrid beamforming antenna structures at both the transmitter and receiver has been considered in~\cite{ref18}. To reduce the system complexity of THz communication, a hybrid beamforming technique based on RIS angle quantization and beam training has been considered in~\cite{ref20,ref22,ref23} and three-dimensional (3D) beamforming using a two-dimensional (2D) planar antenna array has been studied in~\cite{ref19}.  

The previous works mentioned above mainly considered the joint optimization of hybrid beamformers and RIS reflection coefficients assuming that a single RIS is deployed at a fixed position. There are several works considering multiple RISs, which focus on user scheduling and geometry-based user grouping~\cite{ref29,ref31,ref32} and coverage analysis in~\cite{9174910,9868205}. 
For RIS-enabled systems, the optimal placement of RISs is critical to improving  overall system performance, especially for mmWave or THz frequency bands that suffer from multi-path loss, penetration loss, and signal blockage due to obstacles.   
In~\cite{ref27}, a standard setup consisting of a BS and users has been considered to analyze the impact of RIS deployment located near the BS side, near the user side, and two RISs at both sides.  
In spite of the importance of RIS placement optimization to overall system performance, there exist only a few recent works considering the optimal RIS placement for aerial RISs~\cite{ref26, ref33} and for the indoor THz communication setup~\cite{ref21}. 
More specifically, the worst-case SNR maximization over all locations in a target area has been considered for a single aerial RIS in~\cite{ref26} and the sum rate maximization in which a single RIS supports the communication between a BS and multiple users has been considered in~\cite{ref21}.

\subsection{Contribution}
As expected by several 6G vision and requirement reports~\cite{Dang:20,Giordani:20,ref1,TATARIA:21,Koenig:13,Chen:19,Elayan:20}, one of the essential roles of THz communications in 6G systems will be to provide relatively short-range but extremely high data rate services in a crowded indoor environment. However, due to various indoor obstacles such as walls, pillars, furniture, etc.~\cite{ref21,ref24,ref25,ref28}, it is difficult to guarantee unobstructed LoS channels between transmitters and receivers, which become crucially important for ultra-high frequency communications. Therefore, introducing RISs in 6G communication systems will be very effective not only to resolve signal blockage but also to relieve rank-deficient channel conditions and multi-path losses for mmWave or THz communications.   
Motivated by such potentials of RISs, in this paper, we focus on RIS-aided indoor mmWave or THz communication systems in which a BS equipped with a hybrid beamforming antenna system wishes to send independent messages to multiple single-antenna users with the help of multiple RISs. We highlight our contributions as follows:

\begin{itemize}
\item A generalized RIS-aided hybrid beamforming framework incorporated with multiple RISs and multiple users has been considered. Unlike the most previous works, for instance~\cite{ref16,ref18,ref20,ref22,ref23,ref19,ref29,ref31,ref32}, which jointly optimize transmit beamforming and RIS reflecting procedure when the positions of RISs are given, we extend the scope of optimization in this paper to include not only transmit beamforming and RIS reflecting procedure but also the placement of RISs that is crucially important for indoor environment.     
\item In order to efficiently handle such an extended optimization setup, we introduce the RIS placement optimization that is able to maximize the long-term or expected sum rate of multiple users. In particular, we propose a novel construction method of RIS candidate positions reflecting indoor obstacle characteristics and then propose deep learning based selection algorithms searching among the RIS candidate positions with low complexity.  
\item We propose a systematic beam scanning method for constructing the analog beamforming matrix of the BS and the reflection matrices of multiple RISs, which can be implemented from the current 5G hybrid beamforming systems utilizing channel state information (CSI) with reduced size of channel matrices for digital beamforming.
\item Numerical results demonstrate that the proposed two-step optimization, i.e., the RIS placement optimization to maximize the long-term sum rate and then optimization of transmit beamforming and RIS reflecting procedure in real time, is efficient for indoor mmWave or THz communication systems. It is further shown that the proposed analog beamforming and RIS reflecting coefficient construction with reduced CSI can achieve sum rates close to those achievable by coherent beamforming with full CSI.   
\end{itemize}

\subsection{Paper Organization and Notation}

The rest of the paper is organized as follows. In Section \ref{sec:PF}, we describe the problem formulation and introduce the joint optimization of hybrid beamformers and RIS placement and its reflecting procedure.
In Section \ref{sec:scheme_ris_optimization}, we firstly focus on the RIS placement optimization to maximize the coverage area. Then, in Section \ref{sec:scheme_ris_optimization}, a systematic construction of hybrid beamformers and RIS reflecting procedure to maximize the sum rate is considered.
In Section \ref{sec:performance_evaluation}, we numerically evaluate the achievable coverage area and sum rate of the proposed scheme in various aspects and compare them with several benchmark schemes. Finally, Section \ref{sec:conclusion} concludes the paper.

\emph{\bf Notation}: Throughout the paper, vectors and matrices are denoted by bold lowercase and bold uppercase letters, respectively. Let $\|\mathbf{A}\|$, $\mathbf{A}^{\dagger}$, $\mathbf{A}^{-1}$, and $\mathbf{A}^{T}$ denote the Frobenius-norm, the inverse, the complex conjugate transpose, and the transpose of $\mathbf{A}$, respectively. The block diagonal matrix consisting of a set of vectors $\mathbf{a}_1$ to $\mathbf{a}_n$ is denoted by $\operatorname{diag}\{\mathbf{a}_1,\cdots,\mathbf{a}_n\}$. 
The $n\times n$ identity matrix is denoted by $\mathbf{I}_n$. 
For vectors $\mathbf{a}$ and $\mathbf{b}$, $\mathbf{a}\otimes\textbf{b}$ denotes the Kronecker product of $\mathbf{a}$ and $\mathbf{b}$. The $n\times 1$ all-zero vector is denoted by $\mathbf{0}_n$.
For a complex number $a$, the notation $|a|$ and $\angle a$ mean the absolute value and the phase angle of $a$, respectively.  For a finite set $\mathcal{A}$, $|\mathcal{A}|$ means the cardinality of $\mathcal{A}$. Let $[1:n]=\{1,2,\cdots,n\}$ and $\jmath=\sqrt{-1}$. The circularly symmetric complex Gaussian distribution with mean $\mu$ and variance $\sigma^2$ is denoted by $\mathcal{CN}(\mu, \sigma^2)$ and the uniform distribution over $a$ and $b$ is denoted by $\mathcal{U}(a , b)$. The indicator function is denoted by $\mathbf{1}(\cdot)$.

\section{Problem Formulation} \label{sec:PF}
In this section, we state the network and channel models considered in this paper. Then we introduce the sum rate maximization problem for RIS-aided hybrid beamforming systems.

\subsection{Network Model} \label{subsec:Network_Model}
We consider a single-cell indoor downlink cellular network consisting of a BS and $J$ RISs deployed in a 3D space of $\mathcal{S}=[0,S_x] \times [0,S_y]\times [0,S_z]$ depicted in Fig. \ref{Fig 1}. 
The BS wishes to send independent messages to $K$ users with the help of $J$ RISs. 
Denote the positions of the BS, RIS $j$, and user $k$ by $\mathbf{q}_0\in\mathcal{S}$, $\mathbf{q}_j\in\mathcal{S}$, and $\mathbf{q}^{[\text{u}]}_k\in\mathcal{S}$ respectively, where $j\in[1:J]$ and $k\in[1:K]$.
For simplicity, we assume that the BS is located on the ceiling and each RIS is located on one of four side walls.

\begin{figure}[t] \centering 
\includegraphics[scale=0.75]{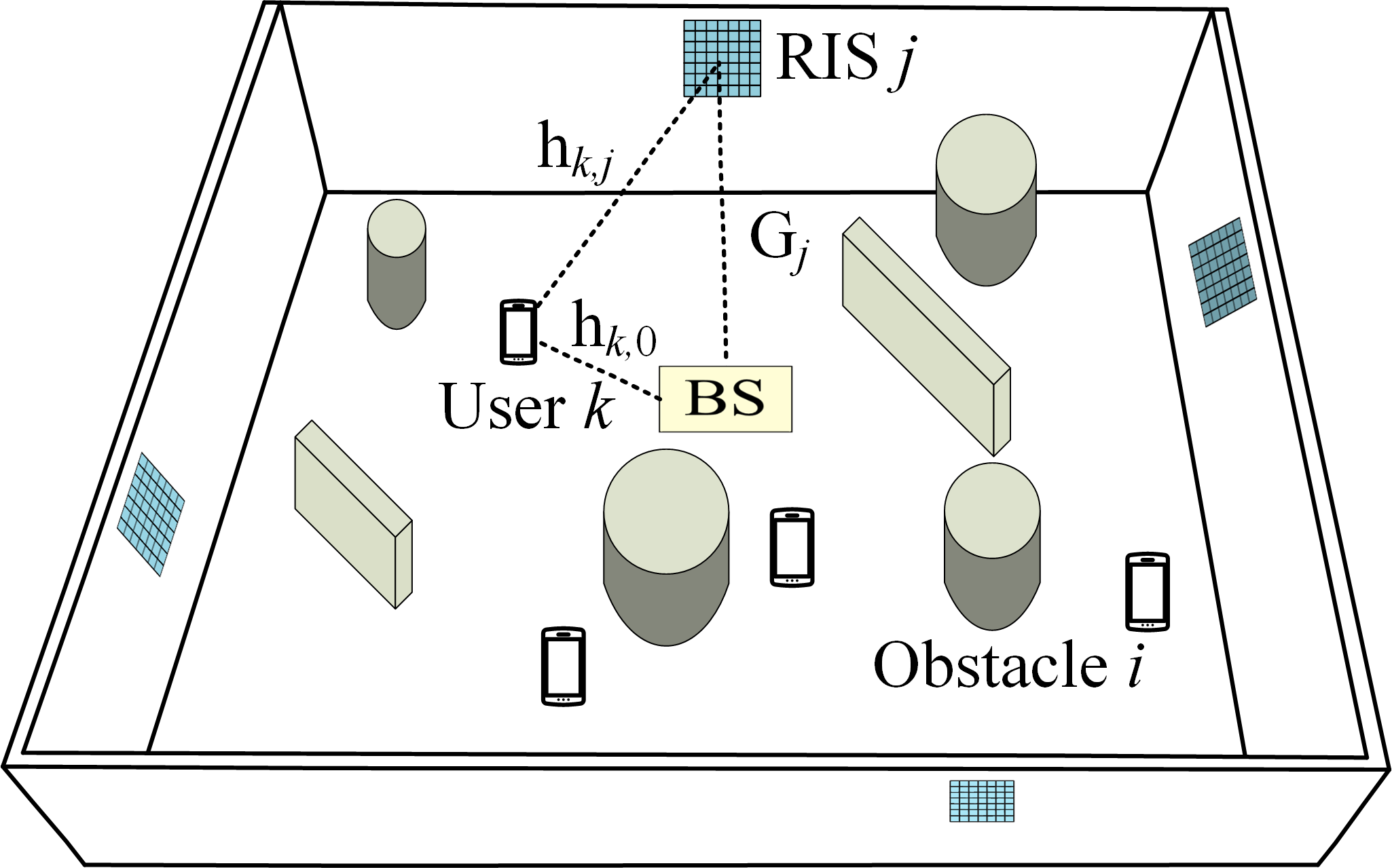} 
\caption{Indoor network model consisting of a BS and multiple RISs.}\label{Fig 1}
\end{figure}

The BS is equipped with a 2D uniform planar hybrid antenna array consisting of $L=L_1L_2$ sub-arrays, each of which is connected to a separate RF chain as illustrated in Fig. \ref{Fig 2}. Under such sub-array antenna structure, each sub-array consists of $M=M_1M_2$ antenna elements deployed in a 2D area. 
Hence, the number of total RF chains and the number of total antenna elements are given by $L$ and $LM$, respectively. Each user is equipped with a single received antenna.
We further assume that each RIS consists of $N=N_1N_2$ reflection elements.
For notational convenience, denote the $(l_1,l_2)$th sub-array by sub-array $l=(l_1-1)L_1+l_2$, where $l_1\in[1:L_1]$ and $l_2\in[1:L_2]$.
Similarly, denote the $(m_1,m_2)$th antenna element in each sub-array by antenna element $m=(m_1-1)M_1+m_2$ and the $(n_1,n_2)$th reflection element in each RIS by reflection element $n=(n_1-1)N_1+n_2$, where $m_1\in[1:M_1]$, $m_2\in[1:M_2]$, $n_1\in[1:N_1]$, and $n_2\in[1:N_2]$.

\begin{remark}
For notational convenience, we will interchangeably use both 2D indices and the corresponding one-dimensional (1D) indices for stating sub-arrays or antenna elements in each sub-array. \hfill$\lozenge$
\end{remark}

\begin{figure}[t] \centering
\includegraphics[scale=0.4]{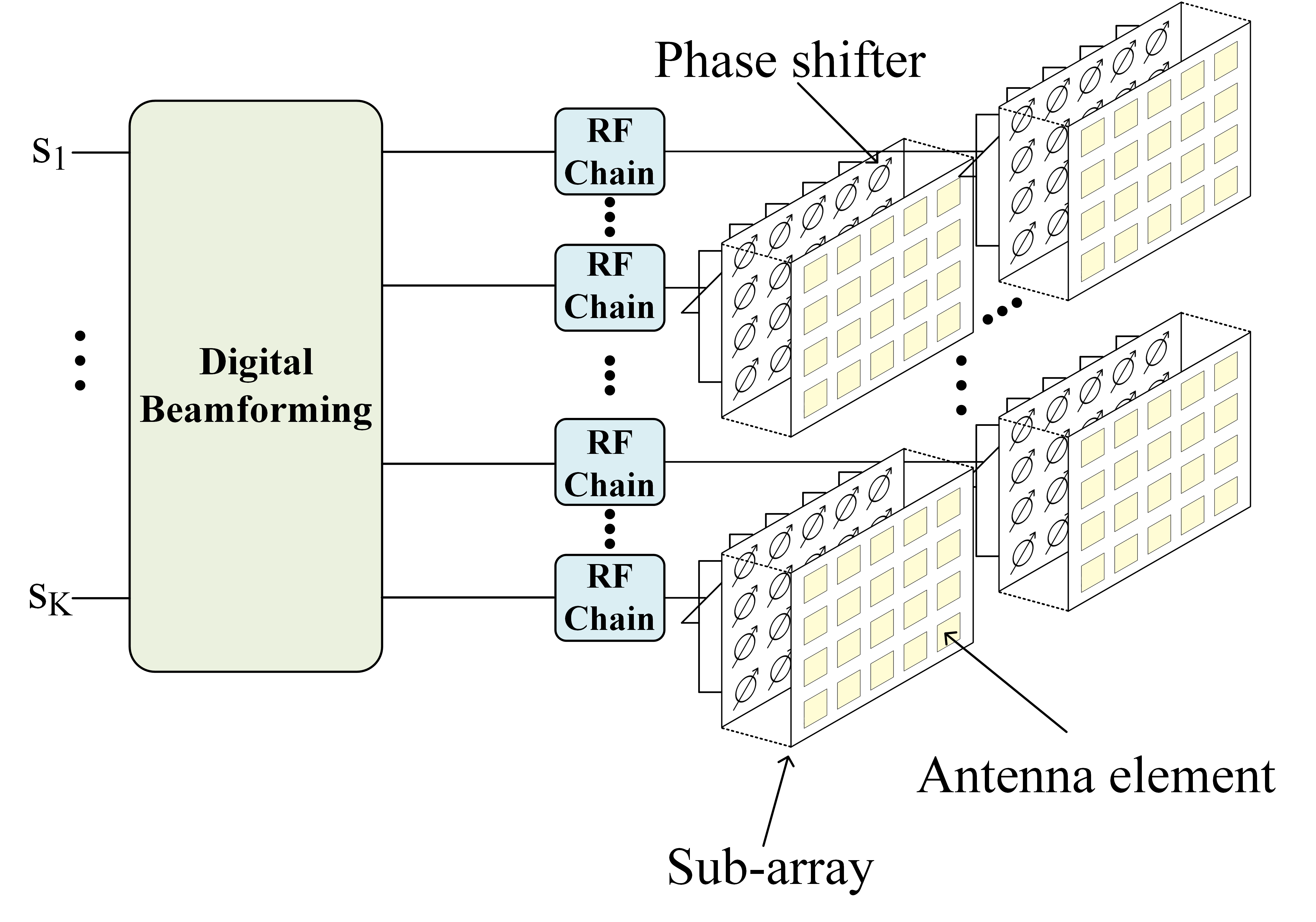}
\caption{Hybrid beamforming antenna system at the BS.}\label{Fig 2}
\end{figure}
In this paper, we mainly focus on massive MIMO environment in which the number of RF chains are greater than or equal to the number of users, i.e., $L\geq K$.
Let $\mathbf{F}_{A}\in \mathbb{C}^{LM\times L}$ be the analog beamforming matrix and $\mathbf{F}_{D}\in \mathbb{C}^{L\times K}$ be the digital beamforming matrix. We assume the sub-array antenna structure at the BS in which each sub-array is connected with a separate RF chain. Hence, $\mathbf{F}_{A}$ is represented by
\begin{align} \label{eq:1_0}
\mathbf{F}_{A}=\operatorname{diag}\left\{\mathbf{f}_{A,1},\mathbf{f}_{A,2},\dots,\mathbf{f}_{A,L}\right\},
\end{align}
where $\mathbf{f}_{A,l} \in \mathbb{C}^{M\times 1}$ is the analog beamforming vector of sub-array $l$ with each element of the form $\frac{1}{\sqrt{M}}e^{\jmath\angle{f_{A,l,m}}}$ for $m\in[1:M]$. Here, $f_{A,l,m}$ denotes the $m$th element of $\mathbf{f}_{A,l}$ satisfying that $|f_{A,l,m}|^2=1$.
The detailed construction of $\mathbf{F}_{A}$ and $\mathbf{F}_{D}$ will be given in Section \ref{sec:scheme_beamtraining}. Denote the transmit signal vector of the BS by $\mathbf{x}\in\mathbb{C}^{LM\times 1 }$, which should satisfy the average power constraint $P$, i.e., $E\left[\|\mathbf{x}\|^2\right]\leq P$.
Let $\mathbf{s}=[s_1,\cdots,s_K]\in \mathbb{C}^{K\times 1}$ denotes the information symbol vector, which follows $E[\mathbf{ss^\dagger}] \leq \frac{{P}}{K}\mathbf{I}_{K}$. That is, $s_k$ is the information symbol for user $k$.   
Denote $\mathbf{F}_D=\left[\mathbf{f}_{D,1},\mathbf{f}_{D,2},\dots,\mathbf{f}_{D,K}\right]$, where $\mathbf{f}_{D,k}\in\mathbb{C}^{L\times 1}$ is the digital beamforming vector for user $k$.
We have 
\begin{align} \label{eq:1}
\mathbf{x}=\mathbf{F}_{A}\mathbf{F}_{D}\mathbf{s}.
\end{align}
Then, to meet the average power constraint, $\|\mathbf{F}_D\|^2 \leq K$ should be satisfied. 
Let $\mathbf{\Sigma}_j \in \mathbb{C}^{N\times N}$ be the reflection matrix of RIS $j$, which is given by
\begin{align}\label{eq:3}
 \mathbf{\Sigma}_j=\operatorname{diag}\{e^{\jmath\sigma_{j,1}},\dots,e^{\jmath\sigma_{j,N}}\},
\end{align}
where $\sigma_{j,n}\in[0,2\pi)$ is the phase reflection coefficient for reflection element $n$ of RIS $j$. The reflection coefficients are assumed to be controlled by the BS through a dedicated controller.

From \eqref{eq:1} and \eqref{eq:3}, the received signal at user $k$ is represented by
\begin{align} \label{eq:2}
y_k=\mathbf{h}_{k,0}\mathbf{x}+\sum_{j=1}^{J}{\mathbf{h}_{k,j}\mathbf{\Sigma}_j\mathbf{G}_{j}}\mathbf{x} +z_k,
\end{align}
where $\mathbf{h}_{k,0}\in\mathbb{C}^{1\times LM}$ is the direct channel vector from the BS to user $k$, $\mathbf{h}_{k,j}\in\mathbb{C}^{1\times N}$ is the channel vector from RIS $j$ to user $k$, $\mathbf{G}_j \in\mathbb{C}^{N\times LM}$ denotes the channel matrix from the BS to RIS $j$, and $z_k$ denotes the additive noise at user $k$, which follows $\mathcal{CN}(0,N_0)$.

\subsection{Indoor mmWave/THz Channel Model} \label{subsec:ch_model}
As illustrated in Fig. \ref{Fig 1}, we consider the indoor environment containing multiple obstacles, which will be explained in detail in Section \ref{sec:scheme_ris_optimization}. Hence, depending on the positions of the BS, RISs, and obstacles, each user might have either the LoS channel or not.  
Furthermore, as the carrier frequency increases, path loss attenuation begins to be more severe. In this case, multi-paths from scattering environments are limited, and rich scattering channel models become inaccurate. For higher frequencies such as mmWave or THz communication, a popular channel model called the Saleh--Valenzuela model is widely adopted \cite{cheng2020low,saleh1987statistical}, which allows us to capture the behavior of channel characteristics at mmWave or THz frequencies. 
Therefore, from \cite{cheng2020low,saleh1987statistical}, the direct channel vector from the BS to user $k$ is given as
\begin{align} \label{eq:5}
\mathbf{h}_{k,0}=&c_{k,0}^{(0)}\mathbf{a}_{\operatorname{BS}}^{(\alpha_0,\beta_0)}\left(\theta_{k,0}^{(t,0)},\phi_{k,0}^{(t,0)}\right)+\frac{1}{\sqrt{Q_{k,0}}}\sum_{q=1}^{Q_{k,0}}c_{k,0}^{(q)}\mathbf{a}_{\operatorname{BS}}^{(\alpha_0,\beta_0)}\left(\theta_{k,0}^{(t,q)},\phi_{k,0}^{(t,q)}\right).
\end{align}
Here the first term denotes the LoS component, represented by channel coefficient $c_{k,0}^{(0)}$, elevation angle $\theta_{k,0}^{(t,0)}$, and azimuth angle $\phi_{k,0}^{(t,0)}$ from the BS  to user $k$. The second term denotes $Q_{k,0}$ Non-LoS (NLoS) components, given channel coefficients $\{c_{k,0}^{(q)}\}_q$, elevation angles $\{\theta_{k,0}^{(t,q)}\}_q$, and azimuth angles $\{\phi_{k,0}^{(t,q)}\}_q$. If there is no LoS due to obstacles, then $c_{k,0}^{(0)}=0$. Otherwise it becomes $c_{k,0}^{(0)}=\frac{1}{\|\mathbf{q}_0-\mathbf{q}^{[\text{u}]}_k\|^{\gamma/2}}$, where $\gamma\geq2$ denotes the path loss exponent. Also, the LoS angles $\theta_{k,0}^{(t,0)}$ and $\phi_{k,0}^{(t,0)}$ are determined by the positions of the BS and user $k$ depicted in Fig. \ref{Fig 3}.
For NLoS components, we assume that $c_{k,0}^{(q)}$, $\theta_{k,0}^{(t,q)}$, and $\phi_{k,0}^{(t,q)}$ follow $\mathcal{CN}({0},\sigma_L^2)$, $\mathcal{U}(- \pi/2 , \pi/2)$, and $\mathcal{U}(- \pi , \pi)$, respectively \cite{cheng2020low,saleh1987statistical}.
\begin{figure}[t] \centering
\includegraphics[scale=0.65]{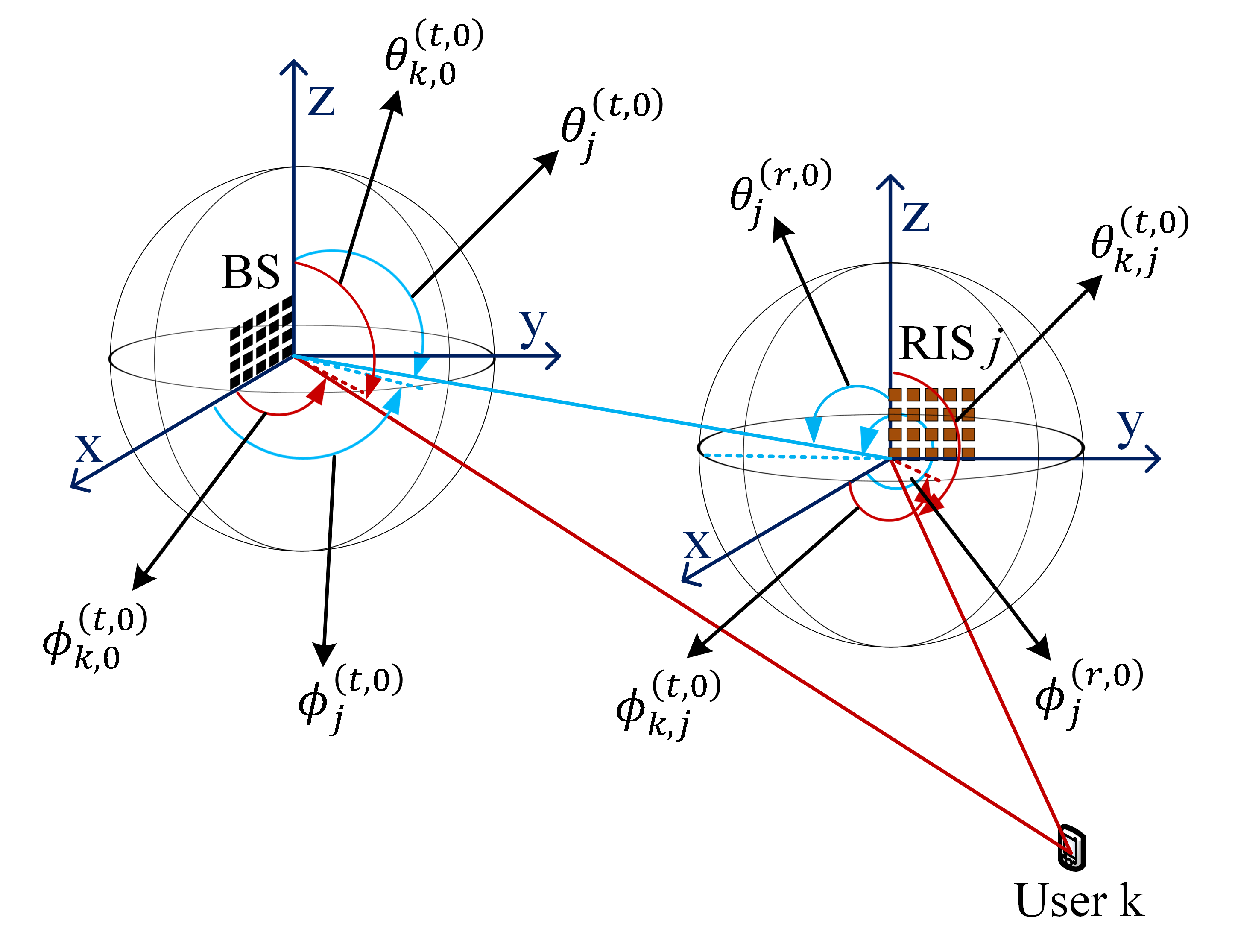}
\caption{Geometries of defining LoS angles for BS--user, RIS--user, and BS--RIS channels.}\label{Fig 3}
\end{figure}
In a similar manner, we define channels from RIS $j$ to user $k$ and channels from the BS to RIS $j$ as

\begin{align} \label{eq:6}
\mathbf{h}_{k,j}=c_{k,j}^{(0)}\mathbf{a}_{\operatorname{RIS}}^{(\alpha_j,\beta_j)}\left(\theta_{k,j}^{(t,0)},\phi_{k,j}^{(t,0)}\right)+\frac{1}{\sqrt{Q_{k,j}}}\sum_{q=1}^{Q_{k,j}}c_{k,j}^{(q)}\mathbf{a}_{\operatorname{RIS}}^{(\alpha_j,\beta_j)}\left(\theta_{k,j}^{(t,q)},\phi_{k,j}^{(t,q)}\right)
\end{align}
and 
\begin{align} \label{eq:7}
\mathbf{G}_{j}&=c_{j}^{(0)}\mathbf{a}_{\operatorname{RIS}}^{(\alpha_j,\beta_j)}\left(\theta_{j}^{(r,0)},\phi_{j}^{(r,0)}\right){\mathbf{a}_{\operatorname{BS}}^{(\alpha_0,\beta_0)}}^\dagger\left(\theta_{j}^{(t,0)},\phi_{j}^{(t,0)}\right) \\ \nonumber &+\frac{1}{\sqrt{Q_j}}\sum_{q=1}^{Q_j}c^{(q)}_{j}\mathbf{a}_{\operatorname{RIS}}^{(\alpha_j,\beta_j)}\left(\theta_{j}^{(r,q)},\phi_{j}^{(r,q)}\right) {\mathbf{a}_{\operatorname{BS}}^{(\alpha_0,\beta_0)}}^\dagger\left(\theta_{j}^{(t,q)},\phi_{j}^{(t,q)}\right),
\end{align}
respectively.

To specify planar array response vectors $\mathbf{a}_{\operatorname{BS}}^{(\alpha_0,\beta_0)}$ and $\mathbf{a}_{\operatorname{RIS}}^{(\alpha_j,\beta_j)}$ in \eqref{eq:5} to \eqref{eq:7}, we introduce a generic form of planar array response vectors as
\begin{align} \label{eq:8}
 \mathbf{a}^{(\alpha,\beta)}\left(\theta,\phi\right) = \mathbf{a}_{{\operatorname{ULA}}}^{(\alpha)}(\theta,\phi) \otimes \mathbf{a}_{{\operatorname{ULA}}}^{(\beta)}(\theta,\phi),
\end{align}
where
\begin{align} \label{eq:9}
\mathbf{a}_{{\operatorname{ULA}}}^{(\alpha)}(\theta,\phi) =\begin{cases}\left[1,e^{\jmath\frac{2\pi}{\lambda}\delta_x\sin{(\theta)}\cos{(\phi)}}, \dots ,e^{\jmath\frac{2\pi}{\lambda}(D_{\operatorname{x}}-1)\delta_x\sin{(\theta)}\cos{(\phi)}}\right]^T &\mbox{ if }\alpha=x,\\
\left[1,e^{\jmath\frac{2\pi}{\lambda}\delta_y\sin{(\theta)}\sin{(\phi)}}, \dots ,e^{\jmath\frac{2\pi}{\lambda}(D_{\operatorname{y}}-1)\delta_y\sin{(\theta)}\sin{(\phi)}}\right]^T &\mbox{ if }\alpha=y,\\
\left[1,e^{\jmath\frac{2\pi}{\lambda}\delta_z\cos{(\theta)}}, \dots ,e^{\jmath\frac{2\pi}{\lambda}(D_{\operatorname{z}}-1)\delta_z\cos{(\theta)}}\right]^T &\mbox{ if }\alpha=z.
\end{cases}
\end{align}
Here, $\lambda$ denotes the wavelength, $D_{\operatorname{x}}$, $D_{\operatorname{y}}$, and $D_{\operatorname{z}}$ are the number of antenna elements along x-, y- and z- directions respectively, and $\delta_{\operatorname{x}}$, $\delta_{\operatorname{y}}$, and $\delta_{\operatorname{z}}$ are the antenna spacing between two adjacent antenna elements along x-, y- and z- directions respectively.
As seen in \eqref{eq:9}, the orientation of the antenna array in the BS or RIS affects to construct $\mathbf{a}_{{\operatorname{ULA}}}^{(\alpha)}(\theta,\phi)$, which can be deployed along the xy-, yz-, or xz-plane. For instance, if the BS antenna array is deployed along the xy-plane, then $\alpha_0=x$ and $\beta_0=y$ in $\mathbf{a}_{\operatorname{BS}}^{(\alpha_0,\beta_0)}$.
In the same manner, we can define $\mathbf{a}_{{\operatorname{ULA}}}^{(\beta)}(\theta,\phi)$ by replacing $\alpha$ in \eqref{eq:9} with $\beta$.

\subsection{Sum Rate Maximization}

For convenience, let 
\begin{align} \label{eq:effective_ch}
\mathbf{w}_{k} =(\mathbf{h}_{k,0}+\sum_{j=1}^{J}{\mathbf{h}_{k,j}\mathbf{\Sigma}_j\mathbf{G}_{j}
})\mathbf{F}_{A} \in  \mathbb{C}^{1\times L }
\end{align}
be the effective channel vector of user $k$ that includes the effect of the analog beamforming at the BS and the reflection processing of $J$ RISs. 
Then, from \eqref{eq:1_0} to \eqref{eq:2}, the received signal at user $k$ is given by
\begin{align} \label{eq:12}
y_k=\mathbf{w}_{k}\mathbf{F}_{D}\mathbf{s} +z_k=\mathbf{w}_{k}\sum_{j=1}^K\mathbf{f}_{D,j}s_{j}.
\end{align}
Hence, the achievable sum rate of $K$ users is represented by
\begin{align}\label{eq:16}
R_{\operatorname{sum}}=\sum_{k=1}^{{K}}\log{\left({1+\frac{P/{K}\|\mathbf{w}_k\mathbf{f}_{D,k}\|^2} {P/{K}\sum_{j=1,j\neq k}^{{K}}\|\mathbf{w}_k\mathbf{f}_{D,j}\|^2+{N_0}}}\right)},
\end{align}
which is a function of $\mathbf{F}_{A}$, $\mathbf{F}_{D},\{\mathbf{\Sigma}_j\}_{j\in[1:J]}$, and $\{\mathbf{q}_j\}_{j\in[1:J]}$.
Note that $R_{\operatorname{sum}}$ depends not only on the analog and digital beamforming at the BS and reflection coefficients at $J$ RISs but also on the set of positions of $J$ RISs, which might be carefully optimized by considering obstacles distributed in the indoor network area. 
Unlike $\mathbf{F}_{A}$, $\mathbf{F}_{D}$, and $\{\mathbf{\Sigma}_j\}_{j\in[1:J]}$, it is quite challenging in practice to adjust $\{\mathbf{q}_j\}_{j\in[1:J]}$ in real time based on the position of users or time-varying channel qualities. Instead, we focus on the RIS placement optimization to maximize the long-term sum rate averaged over the user and channel distributions, whereas $\mathbf{F}_{A}$, $\mathbf{F}_{D}$, and $\{\mathbf{\Sigma}_j\}_{j\in[1:J]}$ are optimized in real time. That is, the sum rate maximization considered in this paper is represented by 
\begin{align} \label{eq:opt_prob}
\max_{\{\mathbf{q}_j\}_{j\in[1:J]}} E\left[\max_{\mathbf{F}_{A},\mathbf{F}_{D},\{\mathbf{\Sigma}_j\}_{j\in[1:J]}}R_{\operatorname{sum}}(\mathbf{F}_{A}, \mathbf{F}_{D},\{\mathbf{\Sigma}_j\}_{j\in[1:J]}; \{\mathbf{q}_j\}_{j\in[1:J]})\right],
\end{align}
where the expectation takes over the distribution of users' position and the corresponding channel distribution.

\section{RIS Placement Optimization} \label{sec:scheme_ris_optimization}

Unlike the conventional beamforming and RIS process design, the optimization problem in \eqref{eq:opt_prob} is quite challenging to solve, mainly because the placement of RISs has to be also optimized considering the distribution of users and the corresponding channels. Hence, in this section, we introduce the coverage area as a simplified performance metric to be used for the RIS placement optimization instead of the expected sum rate given in \eqref{eq:opt_prob}. 
As mentioned in Section \ref{subsec:ch_model}, for mmWave or THz channels, LoS components become dominant and the received signal quality will be severely degraded without LoS components. Therefore, the expected sum rate maximization will be closely related to maximize the LoS region by optimizing the RIS placement assuming that the users are uniformly distributed at random over the entire network area.   

For this purpose, let us first describe a set of $I$ obstacles, in which the signal path from the BS can be blocked by such obstacles. 
We assume the shape of each obstacle is identical over the z-axis as illustrated in Fig. \ref{Fig 1}. Hence the 2D network space $\mathcal{S}=[0,S_x] \times [0,S_y]$ instead of the original 3D space can be introduced for maximizing the coverage region.
To avoid excessive symbol definitions, we will use the same notation used for 3D positions in Section \ref{sec:PF} in this section.
We denote the length of each RIS by $\Delta_{\operatorname{RIS}}$, which is determined by the number of reflection elements and the spacing between two adjacent elements for each RIS, given in \eqref{eq:9}. 
We mainly consider two different types of obstacles; the first is circular obstacles and the second is wall-type obstacles.
For the first model, denote the center position and the radius of obstacle $i$ by $\mathbf{q}_{i}^{[\operatorname{o}]}\in \mathcal{S}$ and $X_i$ respectively for $i\in[1:I]$, where $X_i$ lies in the range $X_l\leq X_i\leq X_u$ and $X_l$ and $X_u$ represent the lower and upper limits respectively.
For the second model, we assume the thickness of walls to be negligible. Therefore, to specify obstacle $i$, denote the center position, length, and angle by $\mathbf{q}_{i}^{[\operatorname{o}]}$, $Y_i$ and $\vartheta_i$ respectively for $i\in[1:I]$. The length $Y_i$ lies in the range $Y_l\leq Y_i\leq Y_u$, where $Y_l$ and $Y_u$ represents the lower and upper limits respectively and the angle is defined counter clockwise with the x-axis in the range $\vartheta_i \in (0,\pi]$.

Let $\mathcal{O}\subseteq \mathcal{S}$ be the region occupied by all the obstacles.
Any arbitrary position in $\mathcal{S}\setminus \mathcal{O}$ is said to be covered if there exists a LoS path from the BS or via one of the RISs. Then the coverage region $\mathcal{C}\left(\{\mathbf{q}_j\}_{j\in[1:J]}\right)$ is defined as the set of all such positions, which is obviously a function of RIS positions $\{\mathbf{q}_j\}_{j\in[1:J]}$.
Denote $C\left(\{\mathbf{q}_j\}_{j\in[1:J]}\right)=\operatorname{Area}(\mathcal{C}\left(\{\mathbf{q}_j\}_{j\in[1:J]}\right))$ by the coverage area attained by $\mathcal{C}\left(\{\mathbf{q}_j\}_{j\in[1:J]}\right)$.
More specifically, assuming the BS as a point source, there exists a LoS path from the BS to the given position if there is no obstacle in the straight line connecting the BS position $\mathbf{q}_0$ and the given position.
Similarly, there exists a LoS path from RIS $j$ to the given position if we can find a specific position in RIS $j$ such that there is no obstacle in the straight line connecting $\mathbf{q}_0$ to the specific position in RIS $j$ and in the straight line connecting the specific position in RIS $j$ to the given position, see Fig. \ref{los_definition} for better understanding.

\begin{figure*}[t]
    \centering
    \subfigure[LoS or NLoS examples.]{
    \includegraphics[width=0.35\columnwidth]{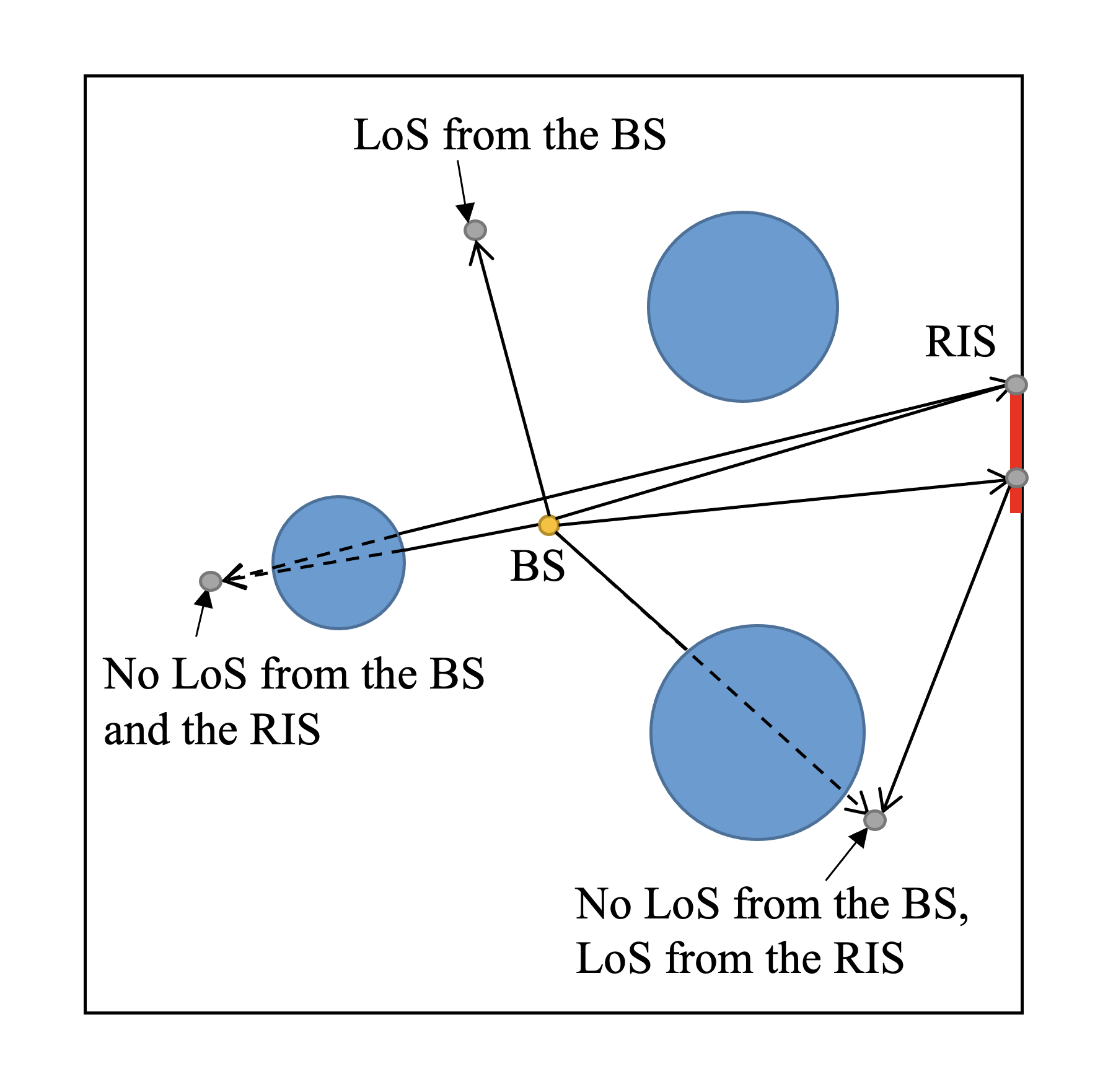}
    \label{los_definition}
    }\quad
    \subfigure[Example of $\mathcal{Q}^*$.]{
    \includegraphics[width=0.35\columnwidth]{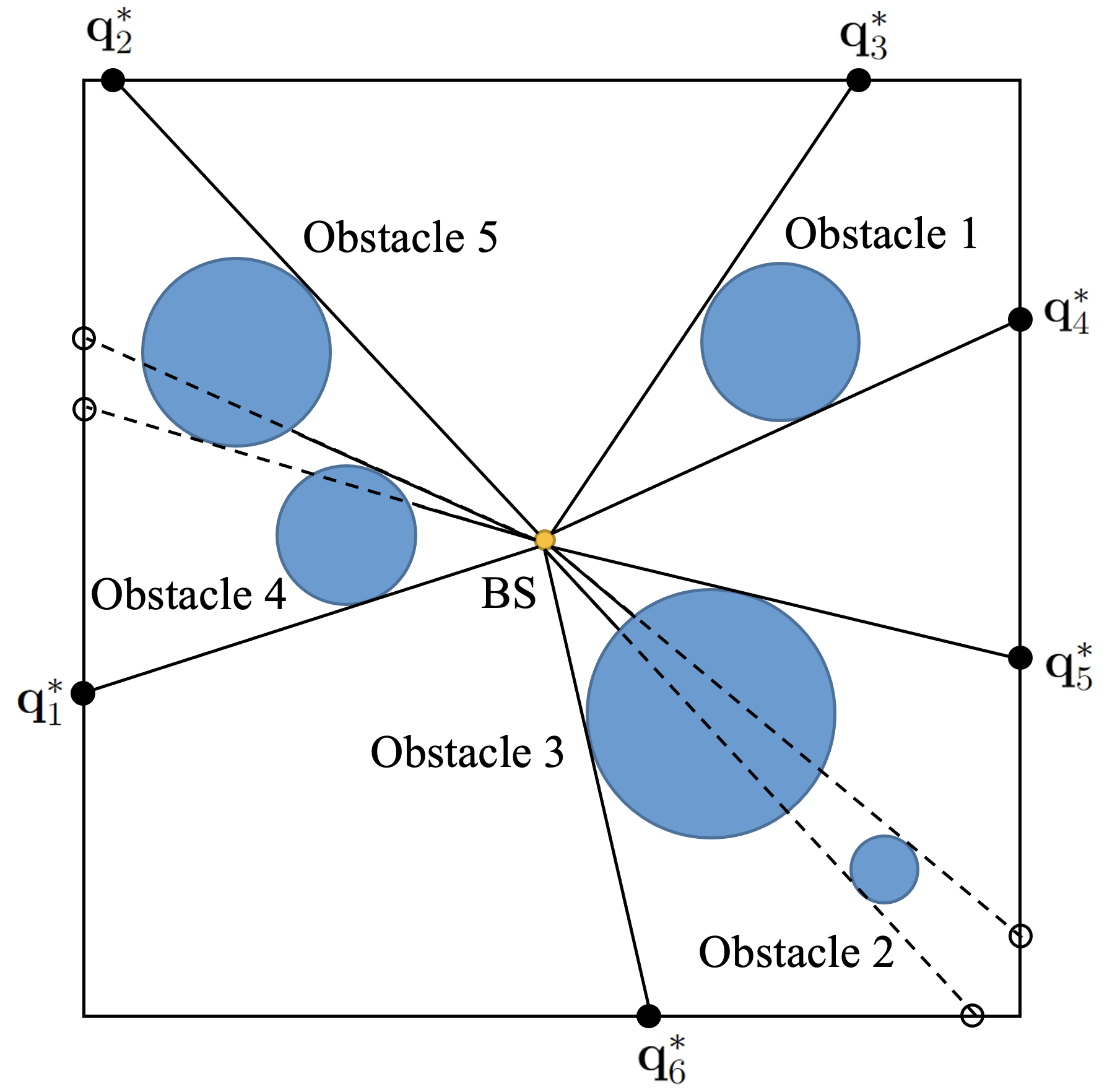}
    \label{q_set}
    }
    \caption{LoS and NLoS for arbitrary positions (a) and the ordered set of $\mathcal{Q}^*$ (b).}
    \label{algorithm_ex}
\end{figure*}

Suppose that the set of possible RIS positions is given as $\mathcal{Q}=\big\{\mathbf{q}^{[\operatorname{RIS}]}_1,\mathbf{q}^{[\operatorname{RIS}]}_2,\cdots,\mathbf{q}^{[\operatorname{RIS}]}_T\big\}$, which is consisted of $T$ candidate positions in four side walls.
Then, the RIS placement optimization under the set $\mathcal{Q}$ is represented as
\begin{align}\label{eq:12_1}
\mathop {\max }\limits_{\mathbf{q}_j \in \mathcal{Q} \text{ for } j\in [1:J],\mathbf{q}_j \neq \mathbf{q}_k \text{ for all } j \neq k }C\left(\{\mathbf{q}_j\}_{j\in[1:J]}\right). 
\end{align}
Denote
\begin{align}\label{eq:14}
\{\mathbf{q}_j^*\}_{j\in[1:J]}=\mathop {\arg\max }\limits_{\mathbf{q}_j \in \mathcal{Q} \text{ for } j\in [1:J],\mathbf{q}_j \neq \mathbf{q}_k \text{ for all } j \neq k }C\left(\{\mathbf{q}_j\}_{j\in[1:J]}\right). 
\end{align}

Note that $\{\mathbf{q}_j^*\}_{j\in[1:J]}$ is optimal when the set $\mathcal{Q}$ is given and, therefore, the candidate set $\mathcal{Q}$ should be also optimized. Obviously, if we equally quantize the positions in four side walls with small enough size and perform \eqref{eq:12_1} under such candidate set, it will provide the optimal RIS placement. 
We propose a deep learning based selection of $\{\mathbf{q}_j^*\}_{j\in[1:J]}$ under the given set $\mathcal{Q}$ in Section \ref{subsec:supervised learning} and then propose an efficient method of constructing the set $\mathcal{Q}$ in Section \ref{subsec:BestRISCandSet}.

\subsection{RIS Placement Optimization via Deep Learning} \label{subsection:RIS_placement}
\label{subsec:supervised learning}
As expected, the coverage or blockage region can be drastically changed depending on the positions and shapes of obstacles. Hence, to efficiently solve the optimization problem in (\ref{eq:12_1}), we apply a deep learning technique which can provide a universally good performance robust to various indoor blockage environments. 
For this purpose, we introduce a feed-forward neural network model consisting of an input layer, output layer, and several hidden layers. We can train this neural network in an offline manner and then utilize the trained model to estimate optimal RIS positions  $\{\mathbf{q}_j^*\}_{j\in[1:J]}$ to provide a reasonably good coverage area in an online manner that significantly reduces computational complexity. 

The input layer consists of the position of the BS $\mathbf{q}_0$ and RIS candidate positions $\big\{\mathbf{q}^{[\operatorname{RIS}]}_t\big\}_{t\in[1:T]}$. To provide information about $I$ obstacles, we also add the center position and radius for circular obstacles and the center position, length, and angle for wall-type obstacles. That is,
$\mathbf{e}_i$ is defined based on the obstacle type as
\begin{align} 
    \mathbf{e}_i= \left\{
        \begin{array}{ll}
            (\mathbf{q}_i^{[o]},X_i) & \quad \text{for the circular obstacle}, \\
            (\mathbf{q}_i^{[o]},Y_i,\vartheta) & \quad \text{for the wall-type obstacle}.
        \end{array}
    \right.
\end{align}
Then the input layer is expressed as 
\begin{align} \label{eq:input_layer}
\left[\mathbf{q}_0, (\mathbf{q}^{[\operatorname{RIS}]}_1,\mathbf{q}^{[\operatorname{RIS}]}_2,\cdots,\mathbf{q}^{[\operatorname{RIS}]}_T), (\mathbf{e}_1,\mathbf{e}_2,\cdots,\mathbf{e}_I)\right].
\end{align}
The output layer consists of $T$ nodes, each of which corresponds to the Q-value of a specific RIS candidate position. To construct the positions of $J$ RISs, we select $J$ largest Q-values among $T$ Q-values and set the positions of $J$ RISs accordingly. 
Denote the resulting RIS positions by $\{\hat{\mathbf{q}}_j\}_{j\in[1:J]}$.

We now explain how to train the neural network. One possible way is to adopt the supervised learning framework, which sets the reward function as one if $(\hat{\mathbf{q}}_{1},\dots,\hat{\mathbf{q}}_{J})=(\mathbf{q}^*_{1},\dots,\mathbf{q}^*_{J})$ and zero otherwise.

Note that the optimal RIS positions $\{\mathbf{q}_j^*\}_{j\in[1:J]}$ is needed in order to calculate this indicator reward. However, the computational complexity attaining $\{\mathbf{q}_j^*\}_{j\in[1:J]}$ via the exhaustive search increases exponentially as the number of candidate positions or the number of RISs increases.
To reduce such computational complexity, we apply the unsupervised learning framework by setting the reward function as the coverage area, i.e., $C\left(\{\hat{\mathbf{q}}_j\}_{j\in[1:J]}\right)$. 

\subsection{Construction of the RIS Candidate Set} \label{subsec:BestRISCandSet}

Previously, we focused on the deep learning based RIS placement optimization when the RIS candidate set $\mathcal{Q}$ is given. In this subsection, we propose an efficient construction method of $\mathcal{Q}$ that can provide a coverage area close to an optimal set while preserving its size reasonably small.

Firstly, a naive approach to construct $\mathcal{Q}$ is to quantize the four side walls with equal size. One advantage of this strategy is that the RIS candidate set $\mathcal{Q}$ is independently given without considering the characteristics of obstacles so that training based on this fixed candidate set $\mathcal{Q}$ might be more stable for various indoor network environments. Furthermore, because $\mathcal{Q}$ is a fixed set,  $\{\mathbf{q}^{[\operatorname{RIS}]}_t\}_{t\in[1:T]}$ in \eqref{eq:input_layer} is not required to provide for constructing the input layer.

Unlike the above approach, we propose a construction method of $\mathcal{Q}$, which depends on obstacle characteristics. Algorithm \ref{alg:construction_Q} states the detailed proposed algorithm. The intuition of the proposed construction is that locating RISs at the edge of LoS and NLoS can enlarge the coverage area for a broad class of indoor environments compared to the cases of locating them in the middle of LoS lines.   
Note that technical issues to be addressed for the proposed algorithm is that the resulting $\mathcal{Q}^*$ can be very different depending on a given indoor environment, which makes training of a neural network hard.
Furthermore, as seen in Algorithm \ref{alg:construction_Q}, the number of RIS candidate positions in $\mathcal{Q}^*$ can be also changed with its maximum number $2I$, i.e., $|\mathcal{Q}^*|\leq 2I$. Recall that $I$ is the number of obstacles in the indoor network.
To handle these technical issues, we firstly construct the input layer and output layer of a neural network assuming $T=2I$.
Then, for training, we rearrange the candidate positions in $\mathcal{Q}^*$ clockwise starting from the $(0,0)$ position and construct the input layer in \eqref{eq:input_layer} based on the reordered set as   
\begin{align}
(\mathbf{q}^{[\operatorname{RIS}]}_1,\mathbf{q}^{[\operatorname{RIS}]}_2,\cdots,\mathbf{q}^{[\operatorname{RIS}]}_T)= (\mathbf{q}^{*}_1,\mathbf{q}^{*}_2,\cdots,\mathbf{q}^{*}_{|\mathcal{Q}^*|}, \mathbf{0},\cdots, \mathbf{0}),
\end{align}
where $\mathbf{q}^{*}_i$ denotes the $i$th position in the ordered set of $\mathcal{Q}^*$.
For better understanding, we refer to the example of the ordered set of $\mathcal{Q}^*$ in Fig. \ref{q_set}.

Because $|\mathcal{Q}^*|\leq 2I$, we add zero vectors for the rest of the input layer.
For the output layer, we choose the largest $J$ Q-values from the first $|\mathcal{Q}^*|$ Q-values among $T$ Q-values and the corresponding $J$ positions are used to $J$ RIS positions. 

\begin{algorithm}[pt]\caption{Proposed construction of $\mathcal{Q}$.} \label{alg:construction_Q}
	\begin{algorithmic}[1]\label{alg0}
		\State {\bf Initialization}: Set $\mathcal{Q}^*=\emptyset$.
		\For {$i\in[1:I]$} 
		\State Draw two straight lines from $\mathbf{q}_0$ tangent to obstacle $i$ from both sides.
 	    \State Denote two positions intersected with the network boundary as $\mathbf{q}'_1$ and $\mathbf{q}'_{2}$.    
            \For {$j\in[1:2]$} 
            \If {No obstacle between the straight line from $\mathbf{q}_0$ to $\mathbf{q}'_j$}
		      \State Update $\mathcal{Q}^*\leftarrow\mathcal{Q}^*\cup\{\mathbf{q}'_j\}$.
		\EndIf		
            \EndFor
  		\EndFor
	\State {\bf Output}: $\mathcal{Q}^*$.
	\end{algorithmic}
\end{algorithm} 

\section{RIS-Aided Hybrid Beamforming Scheme} \label{sec:scheme_beamtraining}

To solve the sum rate maximization in \eqref{eq:opt_prob}, we firstly proposed the efficient method of RIS placement optimization that maximize the coverage area in Section \ref{sec:scheme_ris_optimization}. 
In this section, we propose a systematic construction of hybrid beamforming at the BS and reflection coefficients at multiple RISs to maximize the sum rate in an online manner based on the positions and channels of users.
We apply a codebook-based analog beamforming construction at the BS and a codebook-based construction of reflection coefficients at multiple RISs, each of which are stated in Section \ref{subsec:analog_beamforming} and \ref{subsec:RIS_beam}, respectively. Then the beam scanning and RIS allocation scheme combined with zero-forcing digital beamforming based on beamformed CSI is proposed in Section \ref{subsec:scheme}. 

\subsection{Analog Beamforming at the BS} \label{subsec:analog_beamforming}
Recall that the BS is equipped with a 2D planar hybrid antenna array system having $L$ sub-arrays.  
As each sub-array is capable of steering its analog beam in a certain 3D direction, the analog beamforming of each sub-array in \eqref{eq:1_0} can be constructed as a form of
\begin{align}\label{eq:12}
\mathbf{f}(\theta,\phi)&= \mathbf{a}_{{\operatorname{ULA}}}^{(\alpha_0)}(\theta,\phi) \otimes \mathbf{a}_{{\operatorname{ULA}}}^{(\beta_0)}(\theta,\phi)\nonumber\\
&=\frac{1}{\sqrt{M}}\left[1,e^{-\jmath\frac{2\pi}{\lambda}\delta_z\cos{(\theta)}},e^{-\jmath\frac{2\pi}{\lambda}\delta_y\sin{(\theta)}\sin{(\phi)}},\right.\nonumber\\
&{~~~~~~~~~~~~~~~}\dots ,
\left.e^{-\jmath\frac{2\pi}{\lambda}\left((M_1-1)\delta_y\sin{(\theta)}\sin{(\phi)+(M_2-1)\delta_z\cos{(\theta)}}\right)}\right]^T
\end{align}
when the antenna array of the BS is placed in the yz-plane, i.e., $\alpha_0=x$ and $\beta_0=z$.
See \eqref{eq:8} and \eqref{eq:9} for other cases.

The BS performs codebook-based analog beamforming with the help of RIS reflection procedure, which will be described in detail in Section \ref{subsec:scheme}. In this subsection, we firstly define the analog beamforming codebook used at each of the sub-arrays at the BS.
Let $\{\theta_1,\theta_2,\dots,\theta_{V_1}\}$ and $\{\phi_1,\phi_2,\dots,\phi_{V_2}\}$ be the sets of elevation and azimuth angles used for the codebook construction of the analog beamforming. 
Based on these elevation and azimuth angles, the analog beamforming codebook is constructed as
\begin{align} \label{eq:BS_coodebook}
\mathcal{F}^{[\operatorname{BS}]}=\left\{\mathbf{f}(\theta,\phi) \text{ for all } (\theta,\phi)\in \{\theta_1,\theta_2,\dots,\theta_{V_1}\} \times \{\phi_1,\phi_2,\dots,\phi_{V_2}\}\right\},
\end{align}
where $\mathbf{f}(\theta_{l},\phi_{l})$ is given by \eqref{eq:12}.
In this paper, the same codebook will be used for all sub-arrays of the BS.
That is, the analog beamforming of sub-array $l$ in \eqref{eq:1_0} is given by $\mathbf{f}_{A,l}\in\mathcal{F}^{[\operatorname{BS}]}$ for all $l\in[1:L]$. For notational convenience, denote the $i$th analog beamforming vector in $\mathcal{F}^{[\operatorname{BS}]}$ as $\mathbf{f}^{[\operatorname{BS}]}(i)$, where $i\in[1:V_1V_2]$.

\subsection{RIS Reflection Procedure} \label{subsec:RIS_beam}

For the RIS reflection procedure, each RIS will reflect its received signal from the BS towards to a specific user. Unlike the analog beamforming at the BS, the RIS reflection procedure is also required to consider the incoming signal and LoS channel angles from the BS to each RIS.
Assuming that the analog beamforming at the BS is aligned to the LoS channel angles, the phase reflection coefficient for reflection element $n$ of RIS $j$ in \eqref{eq:3} is required to set as
\begin{align} \label{eq:22}
\sigma_{j,n}\left(\theta_j^{(r,0)},\phi_j^{(r,0)},\theta,\phi\right)=-\frac{2\pi}{\lambda}\left[\delta_\alpha(n_1-1)\left[\Psi^{\alpha_j}+\Psi^{(r,\alpha_{j}})\right]+\delta_\beta(n_2-1)\left[\Psi^{\beta_j}+\Psi^{(r,\beta_{j})}\right]\right]
\end{align}
for all $n\in[1:N]$ in order to reflect its received signal towards to the specific direction of elevation and azimuth angles $(\theta,\phi)$.
Recall that $\theta_j^{(r,0)}$ and$\phi_j^{(r,0)}$ are the elevation and azimuth angles of the LoS component from the BS to RIS $j$ in \eqref{eq:7} and the definitions of $n_1$ and $n_2$ are given in Section \ref{subsec:Network_Model}. Here,
\begin{align} \label{eq:23}
\Psi^{\alpha_j} =\begin{cases}\sin{\theta}\cos{\phi} &\mbox{ if }\alpha_j=x,\\
\sin{\theta}\sin{\phi} &\mbox{ if }\alpha_j=y,\\
\cos{\theta} &\mbox{ if }\alpha_j=z
\end{cases}
\end{align}
and
\begin{align} \label{eq:24}
\Psi^{(r,\alpha_{j})} =\begin{cases}\sin{\theta}_j^{(r,0)}\cos{\phi}_j^{(r,0)} &\mbox{ if }\alpha_j=x,\\
\sin{\theta}_j^{(r,0)}\sin{\phi}_j^{(r,0)} &\mbox{ if }\alpha_j=y,\\
\cos{\theta}_j^{(r,0)} &\mbox{ if }\alpha_j=z.
\end{cases}
\end{align}
As the same manner, $\Psi^{\beta_j}$ and $\Psi^{(r,\beta_{j})}$ can be defined.
Then, denote $\mathbf{\Sigma}_j(\theta_j^{(r,0)},\phi_j^{(r,0)},\theta,\phi)$ as the reflection matrix of RIS $j$ consisting of the phase reflection coefficients in \eqref{eq:22}.

Similar to the analog beamforming at the BS, each RIS performs codebook-based RIS reflection. 
We define the angle pair $\theta \in \{\theta_1^{[\operatorname{RIS}]},\theta_2^{[\operatorname{RIS}]},\dots,\theta_{W_1}^{[\operatorname{RIS}]}\}$ and $\phi \in \{\phi_1^{[\operatorname{RIS}]},\phi_2^{[\operatorname{RIS}]},\dots,\phi_{W_2}^{[\operatorname{RIS}]}\}$. Then the reflection codebook of RIS $j$ is constructed as
\begin{align} \label{eq:ris_codebook}
\mathcal{F}_j^{[\operatorname{RIS}]}= \left\{\mathbf{\Sigma}_j(\theta_j^{(r,0)},\phi_j^{(r,0)},\theta,\phi) \text{ for all } (\theta,\phi) \in \{\theta_1^{[\operatorname{RIS}]},\dots,\theta_{W_1}^{[\operatorname{RIS}]}\}\times\{\phi_1^{[\operatorname{RIS}]},\dots,\phi_{W_2}^{[\operatorname{RIS}]}
\}\right\}.
\end{align}
As seen in \eqref{eq:ris_codebook}, because the LoS angles from the BS are different for each RIS, the resulting reflection codebook is different for each RIS even if each RIS utilizes the same set of elevation and azimuth angles for reflection.
For notational convenience, denote the $i$th reflection matrix in $\mathcal{F}_j^{[\operatorname{RIS}]}$ as $\mathbf{F}^{[\operatorname{RIS}]}_j(i)$, where $i\in[1:W_1W_2]$.

\begin{remark}
In this paper, we assume that the BS--user channels $\{\mathbf{h}_{k,0}\}_k$ and the RIS--user channels $\{\mathbf{h}_{k,j}\}_{k,j}$ are not available at the BS. Hence, the codebook-based beamforming is required. On the other hand, the LoS angles from the BS to RISs $\{\theta_j^{(r,0)},\phi_j^{(r,0)}\}_{j}$ is assumed to be available and therefore can be used for constructing $\{\mathcal{F}_j^{[\operatorname{RIS}]}\}_j$ in \eqref{eq:ris_codebook}. \hfill$\lozenge$
\end{remark}

\subsection{Hybrid Beamforming and RIS Procedure via Beam Scanning} \label{subsec:scheme}
In this subsection, we state the BS beam scanning based on $\mathcal{F}^{[\operatorname{BS}]}$ in \eqref{eq:BS_coodebook} and the RIS beam scanning based on $\mathcal{F}_j^{[\operatorname{RIS}]}$ in \eqref{eq:ris_codebook} for RIS $j$, which will be sequentially conducted.
Then, each user $k$ will report the measured SNR to the BS. Denote $\operatorname{SNR}_k(i)$ as the received SNR of user $k$ when the BS sends $\mathbf{f}^{[\operatorname{BS}]}(i)$, where $i\in[1:V_1V_2]$.
Similarly, denote $\operatorname{SNR}_k^{[\operatorname{RIS}]}(j,i)$ as the received SNR of user $k$ when RIS $j$ sets $\mathbf{\Sigma}_j=\mathbf{F}_j^{[\operatorname{RIS}]}(i)$, where $i\in[1:W_1W_2]$. The detailed beam scanning and SNR reporting are given in Algorithm \ref{alg:beamscanning}. Lines 1 to 4 are the BS beam scanning and the corresponding SNR reporting.
Lines 6 to 10 are RIS $j$ beam scanning while steering the BS beam to RIS $j$. In Line 7, the definitions of $(\theta_j^{(t,0)},\phi_j^{(t,0)})$ and $\mathbf{f}(\theta,\phi)$ are given in \eqref{eq:7} and \eqref{eq:12}, respectively.

Note that the first sub-array of the BS is only activated for the BS and RIS beam scanning in Algorithm \ref{alg:beamscanning}, see Lines 2 and 7.
The following lemma shows that the received SNR of user $k$ is the same by activating one of $L$ sub-arrays of the BS if the BS--user channel only contains the LoS component. Hence, for the BS beam scanning procedure, the first sub-array is only used.

\begin{algorithm}[pt]\caption{Beam scanning and SNR reporting.} \label{alg:beamscanning}
	\begin{algorithmic}[1]\label{alg2}
\For {$i\in{\left[1:V_1V_2\right]}$} 
		\State The BS transmits $\mathbf{x}=\left[\mathbf{f}^{[\operatorname{BS}]}(i),\mathbf{0}_M^T,\dots,\mathbf{0}_M^T\right]^T$.
  \State Each user $k$ reports $\operatorname{SNR}_k(i)$.
\EndFor
\For {$j\in{[1:J]}$}
\For {$i\in{[1:W_1W_2]}$}
  		\State The BS transmits $\mathbf{x}=\left[\mathbf{f}(\theta_j^{(t,0)},\phi_j^{(t,0)}),\mathbf{0}_M^T,\dots,\mathbf{0}_M^T\right]^T$.
   \State RIS $j$ sets $\mathbf{\Sigma}_j=\mathbf{F}_j^{[\operatorname{RIS}]}(i)$.
  \State Each user $k$ reports $\operatorname{SNR}_k^{[\operatorname{RIS}]}(j,i)$.
\EndFor 
\EndFor
\State {\bf{ Output}}: $\{\operatorname{SNR}_k(i)\}_{k\in{[1:K]}, i\in{[1:V_1V_2]}}$, $\{\operatorname{SNR}_k^{[\operatorname{RIS}]}(j,i)\}_{k\in{[1:K]},j\in{[1:J]}, i\in{[1:W_1W_2]}}$.
	\end{algorithmic}
\end{algorithm}

\begin{lemma}\label{th:2} 
Denote $\mathbf{F}_{A,l}\left(\theta,\phi\right)$ as $\mathbf{F}_{A}$ when $\mathbf{f}_{A,l'}=\mathbf{f}(\theta,\phi)$ for $l'=l$ and $\mathbf{f}_{A,l'}=\mathbf{0}$ otherwise in \eqref{eq:1_0}.
Then, for any arbitrary $\theta$ and $\phi$, $\|c_{k,0}^{(0)}\mathbf{a}_{\operatorname{BS}}^{(\alpha_0,\beta_0)}\left(\theta_{k,0}^{(t,0)},\phi_{k,0}^{(t,0)}\right)\mathbf{F}_{A,l}\left(\theta,\phi\right)\|^2$ is the same for all $l\in[1:L]$, where $c_{k,0}^{(0)}\mathbf{a}_{\operatorname{BS}}^{(\alpha_0,\beta_0)}\left(\theta_{k,0}^{(t,0)},\phi_{k,0}^{(t,0)}\right)$ is the LoS channel component in \eqref{eq:5}.
\end{lemma}
\begin{proof}
For simplicity, we assume that the BS is placed in the yz-plane, i.e., $\alpha_0=y$ and $\beta_0=z$. We have
\begin{align}\label{eq:29}
&\|c_{k,0}^{(0)}\mathbf{a}_{\operatorname{BS}}^{(\alpha_0,\beta_0)}\left(\theta_{k,0}^{(t,0)},\phi_{k,0}^{(t,0)}\right)\mathbf{F}_{A,l}\left(\theta,\phi\right)\|^2\nonumber\\
&=\bigg\|\frac{c_{k,0}^{(0)}}{\sqrt{M}}\Big[0,\dots,\big(e^{\jmath\frac{2\pi}{\lambda}\left((l_1-1)\delta_y\sin{\theta_{k,0}^{(t,0)}}\sin{\phi_{k,0}^{(t,0)}+(l_2-1)\delta_z\cos{\theta_{k,0}^{(t,0)}}}\right)}+\nonumber\\
&{~~~~~}\dots+e^{\jmath\frac{2\pi}{\lambda}\left((l_1-1)(M_1-1)\delta_y\sin{\theta_{k,0}^{(t,0)}}\sin{\phi_{k,0}^{(t,0)}+(l_2-1)(M_2-1)\delta_z\cos{\theta_{k,0}^{(t,0)}}}\right)}\times\nonumber \\
&{~~~~~}e^{-\jmath\frac{2\pi}{\lambda}\left((M_1-1)\delta_y\sin{\theta}\sin{\phi-(M_2-1)\delta_z\cos{\theta}}\right)}\big) ,\dots,0\Big]\bigg\|^2, \nonumber \\
&=\bigg\|\frac{c_{k,0}^{(0)}}{\sqrt{M}}e^{\jmath\frac{2\pi}{\lambda}\left((l_1-1)\delta_y\sin\theta_{k,0}^{(t,0)}\sin\phi_{k,0}^{(t,0)}+(l_2-1)\delta_z \cos\theta_{k,0}^{(t,0)}\right)}\Big[0,\dots,\big(1+\nonumber\\
&{~~~~~}\dots+e^{\jmath\frac{2\pi}{\lambda}\left((M_1-1)\delta_y\sin{\theta_{k,0}^{(t,0)}}\sin{\phi_{k,0}^{(t,0)}+(M_2-1)\delta_z\cos{\theta_{k,0}^{(t,0)}}}\right)}\times\nonumber \\
&{~~~~~}e^{-\jmath\frac{2\pi}{\lambda}\left((M_1-1)\delta_y\sin{\theta}\sin{\phi-(M_2-1)\delta_z\cos{\theta}}\right)}\big) ,\dots,0\Big]\bigg\|^2, \nonumber \\
&=\|c_{k,0}^{(0)}\mathbf{a}_{\operatorname{BS}}^{(\alpha_0,\beta_0)}\left(\theta_{k,0}^{(t,0)},\phi_{k,0}^{(t,0)}\right)\mathbf{F}_{A,1}\left(\theta,\phi\right)\|^2,
\end{align}
which completes the proof.
\end{proof}

Based on the reported SNR values from all users, each user $k$ will be served directly from the BS or be served from one of the $J$ RISs.
To state such sub-array and RIS assignment, we define the assignment vector $\mathbf{u}=\left[u_1,u_2,\dots,u_K\right]$ in which $u_k\in[0:K]$ denotes the assignment of user $k$ such that $u_k=0$ means that user $k$ will be served directly from the BS and $u_k=j$ means that user $k$ will be served from RIS $j$.
Also, the following conditions should be satisfied.
\begin{align} \label{eq:assign_condition}
    \sum_{k\in[1:K]} \mathbf{1}(u_k=j)\leq 1 {~}\text{for all}{~} j\in[1:J], \nonumber \\
    \sum_{j\in[1:J]} \mathbf{1}(u_k=j)\leq 1 {~}\text{for all}{~} k\in[1:J].
\end{align}

Algorithm \ref{alg:user_selection} states the proposed algorithm to set the assignment vector $\mathbf{u}$ when SNR values are reported.
Note that the resulting assignment vector $\mathbf{u}^*$ from Algorithm \ref{alg:user_selection} satisfies the conditions in \eqref{eq:assign_condition}.
Recall that we focus on the regime that $L\geq K$. Hence, if $u^*_k=0$, one of the $L$ sub-arrays at the BS can be assigned to user $k$ and its analog beamforming can be set as the maximum beam index from the reported $\{\operatorname{SNR}_k(i)\}_{i\in{[1:V_1V_2]}}$.  
If $u^*_k\neq 0$, one of the $L$ sub-arrays at the BS can be assigned to user $k$ and its analog beamforming can be set to the direction of RIS $u^*_k$, see Line 7 in Algorithm \ref{alg:beamscanning}. Then RIS $u^*_k$ sets its reflection matrix as the maximum reflection index from the reported $\{\operatorname{SNR}_k^{[\operatorname{RIS}]}(u^*_k,i)\}_{i\in{[1:W_1W_2]}}$.
For non-assigned sub-arrays and RISs, we set the corresponding sub-arrays and RISs as all-zero vectors or matrices.
From the procedure above, the analog beamforming matrix $\mathbf{F}_{A}$ in \eqref{eq:1_0} and the reflection matrix $\mathbf{\Sigma}_j$ in \eqref{eq:3} for all $j\in[1:J]$ can be constructed.
Finally, to state the digital beamforming matrix, denote $\mathbf{W}=\left[\mathbf{w}^T_1,\cdots,\mathbf{w}^T_K\right]^T$, which is the concatenated matrix consisting of the effective channel vectors of $K$ users in \eqref{eq:effective_ch}. 
We apply a generalized zero-forcing beamforming by setting
\begin{align} \label{eq:digtal_beamforming}
\mathbf{f}_{D,k}=\frac{[\mathbf{W}^{\dagger}(\mathbf{W}\mathbf{W}^{\dagger})^{-1}]_k}{\|[\mathbf{W}^{\dagger}(\mathbf{W}\mathbf{W}^{\dagger})^{-1}]_k\|}
\end{align}
for all $k\in[1:K]$, where $[\mathbf{A}
]_k$ denotes the $i$th column vector of the matrix $\mathbf{A}$.
In the end, the achievable sum rate of the proposed scheme can be calculated from \eqref{eq:16}.  

\begin{remark}
To construct the digital beamforming vector in \eqref{eq:digtal_beamforming}, each user $k$ is only required to feedback its effective channel vector $\mathbf{w}_k$ of size $L$ to the BS instead of the original channel vector $\mathbf{h}_{k,0}$ of size $LM$. Furthermore, the end-to-end channel estimation from the BS to each user is enough instead of measuring BS--user, BS--RIS, and RIS--user channels separately.  \hfill$\lozenge$
\end{remark}

\begin{algorithm}[pt]\caption{Sub-array and RIS assignment.} \label{alg:user_selection}
	\begin{algorithmic}[1]\label{alg3}
		\State {\bf Input}:  $\{\operatorname{SNR}_k(i)\}_{k\in{[1:K]}, i\in{[1:V_1V_2]}}$, $\{\operatorname{SNR}_k^{[\operatorname{RIS}]}(j,i)\}_{k\in{[1:K]},j\in{[1:J]}, i\in{[1:W_1W_2]}}$.
		\State {\bf Initialization}: Set $\mathcal{K}=[1:K]$ and $\mathcal{J}=[1:J]$.
		\For {$k\in[1:K]$} 
		\State Calculate $\mathop{\max}\limits_{k'\in \mathcal{K}, i\in[1:V_1V_2]}\operatorname{SNR}_{k'}(i)\triangleq\operatorname{SNR}_{\text{max}} $.
  		\State Calculate {$\mathop{\max}\limits_{k'\in \mathcal{K}, j\in \mathcal{J},  i\in[1:W_1W_2]}\operatorname{SNR}^{[\operatorname{RIS}]}_{k'}(j,i)\triangleq\operatorname{SNR}_{\text{max}}^{[\operatorname{RIS}]} $}.
            \If {$\operatorname{SNR}_{\text{max}} \geq \operatorname{SNR}_{\text{max}}^{[\operatorname{RIS}]}$}
		      \State  Set $u^*_{\bar{k}}=0$, where $ \bar{k}=\mathop{\arg \max}\limits_{k'\in \mathcal{K}, i\in[1:V_1V_2]}\operatorname{SNR}_{k'}(i)$.
        \State Update $\mathcal{K}\leftarrow\mathcal{K}\setminus\{\bar{k}\}$.
        \Else
        	      \State  Set $u^*_{\bar{k}}=\bar{j}$, where $ (\bar{k},\bar{j})=\mathop{\arg\max}\limits_{k'\in \mathcal{K},j\in \mathcal{J}, i\in[1:W_1W_2]}\operatorname{SNR}^{[\operatorname{RIS}]}_{k'}(j,i) $.
        \State Update $\mathcal{K}\leftarrow\mathcal{K}\setminus\{\bar{k}\}$ and $\mathcal{J}\leftarrow\mathcal{J}\setminus\{\bar{j}\}$.
        
		\EndIf		  
            \EndFor
\State {\bf Output}: $\mathbf{u}^*=\left[u^*_1,u^*_2,\dots,u^*_K\right]$.
	\end{algorithmic}
\end{algorithm} 

\section{Numerical Evaluation and Discussion} \label{sec:performance_evaluation}

In this section, we evaluate the coverage area and sum rate achievable by the proposed scheme and compare them with several benchmark schemes. 

\subsection{Evaluation of Coverage Area} \label{subsec:coverage}

As mentioned in Section \ref{sec:scheme_ris_optimization}, the original 3D indoor network is converted into the corresponding 2D indoor network assuming LoS only channels for the RIS placement optimization. 
Hence from the 3D indoor network environment in Table \ref{tab:environment_3d}, the 2D network region is given by $\mathcal{S}=[0,10{~}\text{m}]\times[0,10{~}\text{m}]$ and the BS position is given by $\mathbf{q}_0=[5, 5]$.
For circular obstacles, we assume that $\mathbf{q}^{[\operatorname{o}]}_i$ is drawn uniformly at random from $\mathcal{S}$ and $X_i$ is drawn from $\mathcal{U}(0.5,1.5)$. For wall-type obstacles, we assume that $\mathbf{q}^{[\operatorname{o}]}_i$ is drawn uniformly at random from $\mathcal{S}$ and $Y_i$ and $\vartheta_i$ are drawn from $\mathcal{U}(1,7.3)$ and $\mathcal{U}(0,\pi)$ respectively.
Also, Table \ref{tab:hyperparameters} summarizes the main hyperparameters for the proposed supervised and unsupervised learning.  

\begin{table}\caption{Hyperparameters for deep learning.}\label{tab:hyperparameters}
	\begin{center}
\scalebox{0.85}{
			\begin{tabular}{m{5cm}<{\centering} |m{5cm}<{\centering}}
				\hline
				\multicolumn{1}{c|}{\bf {Parameters}} & \multicolumn{1}{c}{\bf {Values}}\\
				\hline
				\hline
				{Learning rate} & {$0.001$}\\
				\hline
				{Activation function} & {Relu}\\
				\hline
				{Optimizer} & {Adam}\\
				\hline
				{Reward for supervised learning} & {Indicator function {(see Section \ref{subsection:RIS_placement})}} \\
				\hline
	                          {Reward for unsupervised learning} & {Coverage area {(see Section \ref{subsection:RIS_placement})}} \\
				\hline		    
                 {Discount factor} & {$0.95$} \\
			    \hline
				{$\epsilon$ for $\epsilon$-greedy} & {$0.01$} \\
				\hline
				{Minibatch size} & {$32$} \\
				\hline
				{Number of hidden layers} & {$3$} \\
				\hline
				{Number of nodes in hidden layers} & {$64,128,64$} \\
				\hline
		\end{tabular}}
	\end{center}
\end{table}

\subsubsection{Benchmark schemes and computational complexity}
For comparison, we consider several methods for the construction of the RIS positions $\{\mathbf{q}_j\}_{j\in[1:J]}$.  
In this subsection, `Optimal' means the full search algorithm over the set $\mathcal{Q}$ constructing with equal-size quantization of the entire network boundary and letting the quantization size small enough, which is a performance upper bound. 
Also, `Full search in $\mathcal{Q}^*$', `SL in $\mathcal{Q}^*$', and `UL in $\mathcal{Q}^*$' mean the full search, the proposed supervised learning, and the proposed unsupervised learning over the set $\mathcal{Q}^*$, respectively. Lastly, `Random' means that RISs are deployed randomly over the network boundary, which can be regarded as a performance lower bound.  
If we quantize the network boundary with size of $\delta$, then the number of RIS candidate positions is given as $\frac{2S_xS_y}{\delta}$, which linearly increases with decreasing $\delta$. Whereas, the number of RIS candidate positions in $\mathcal{Q}^*$ is upper bounded by $2I$, where $I$ is the number of obstacles.
Even though the RIS placement optimization in $\mathcal{Q}^*$ can significantly reduce the searching space, the number of possible combinations increases exponentially with the number RISs $J$ increases. Hence, for the cases of `Full search in $\mathcal{Q}^*$' and `SL in $\mathcal{Q}^*$', computational complexity increases exponentially with $J$, which might be not applicable for large $J$. Note that for `SL in $\mathcal{Q}^*$', such computational complexity is only required for the training phase.
Lastly, for `UL in $\mathcal{Q}^*$', computational complexity does not increase exponentially with $J$ so that is can be applied for large $J$.

\subsubsection{Comparison of coverage areas}
Fig. \ref{coverage_ex} illustrates exemplary coverage regions with and without RISs, where the positions of RISs are optimized based on `UL in $\mathcal{Q}^*$'. 
In the figures, the blue regions correspond to the regions covered from the BS and the red regions correspond to the regions covered by RISs.
It is seen that few RISs with optimized placement can efficiently resolve the signal blockage for indoor environment. 
It is further shown in Section \ref{subsec:sum_rate} that such improvement on coverage regions provides improved sum rates when users are distributed over the network area.

\begin{figure*}[t]
    \centering
    \subfigure[No RIS (circular).]{
    \includegraphics[width=0.2\columnwidth]{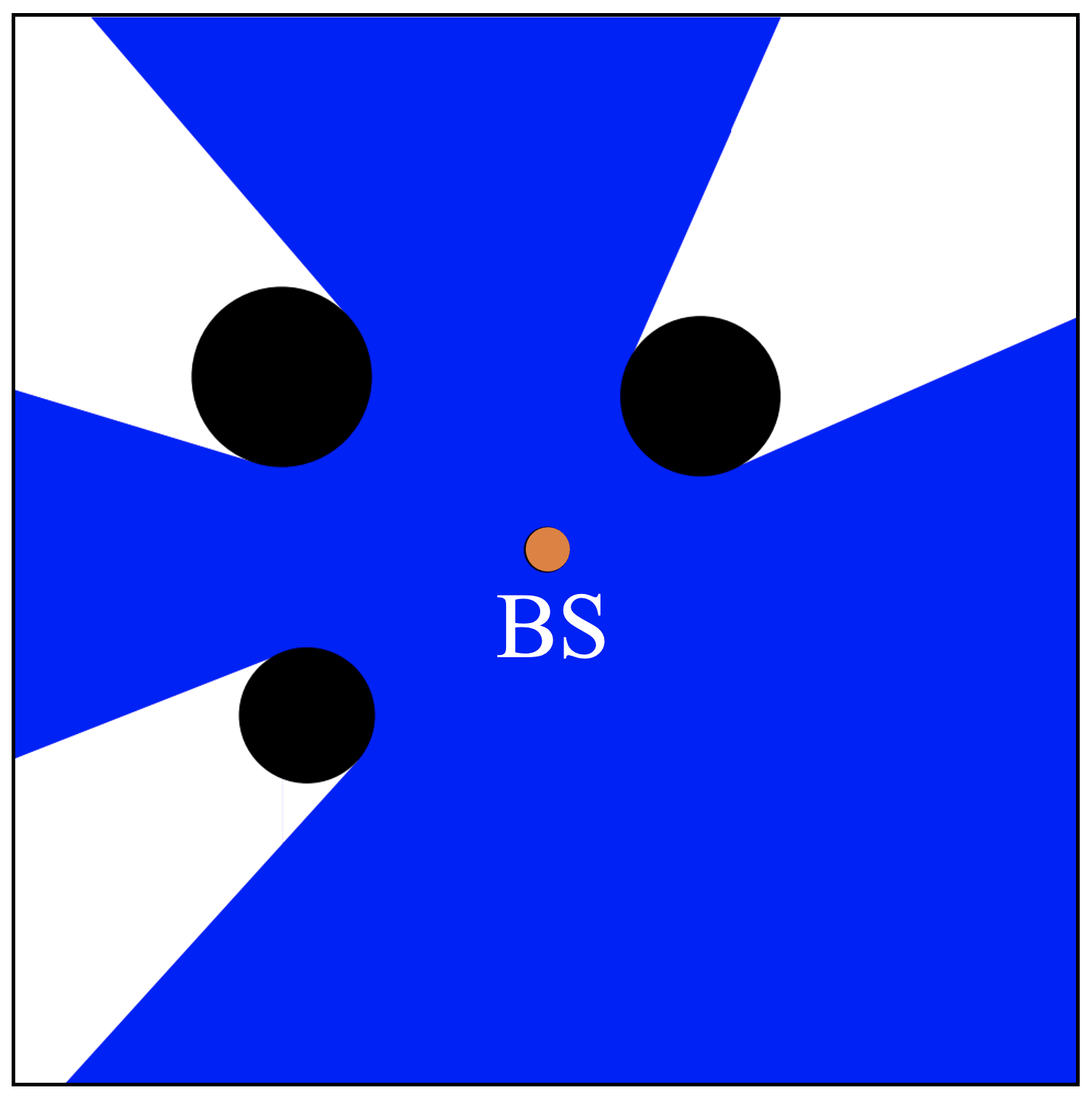}
    \label{coverage_ex1}
    }\quad
    \subfigure[$2$ RISs (circular).]{
    \includegraphics[width=0.2\columnwidth]{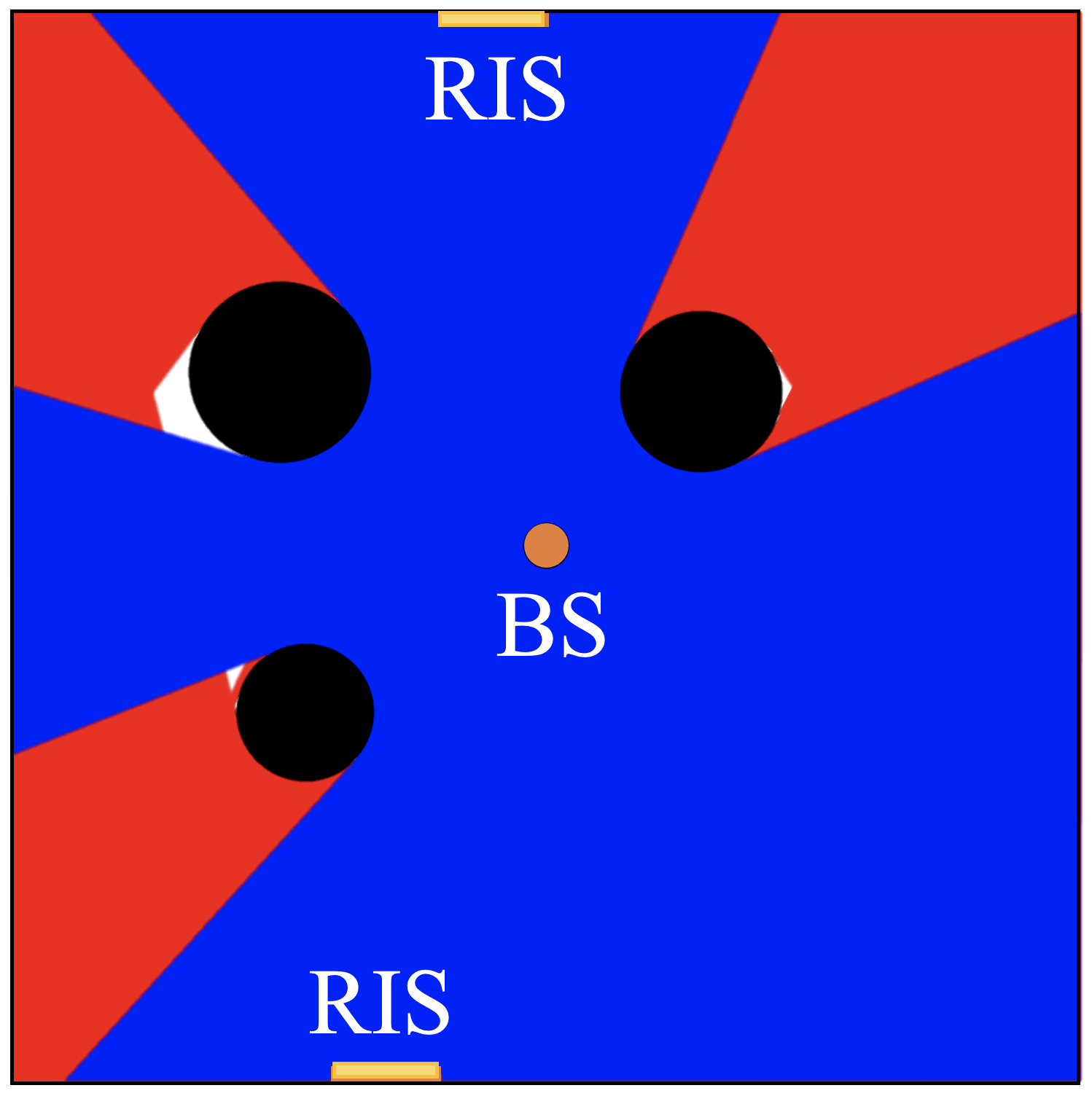}
    \label{coverage_ex2}
    }
    \subfigure[No RIS (wall-type).]{
    \includegraphics[width=0.2\columnwidth]{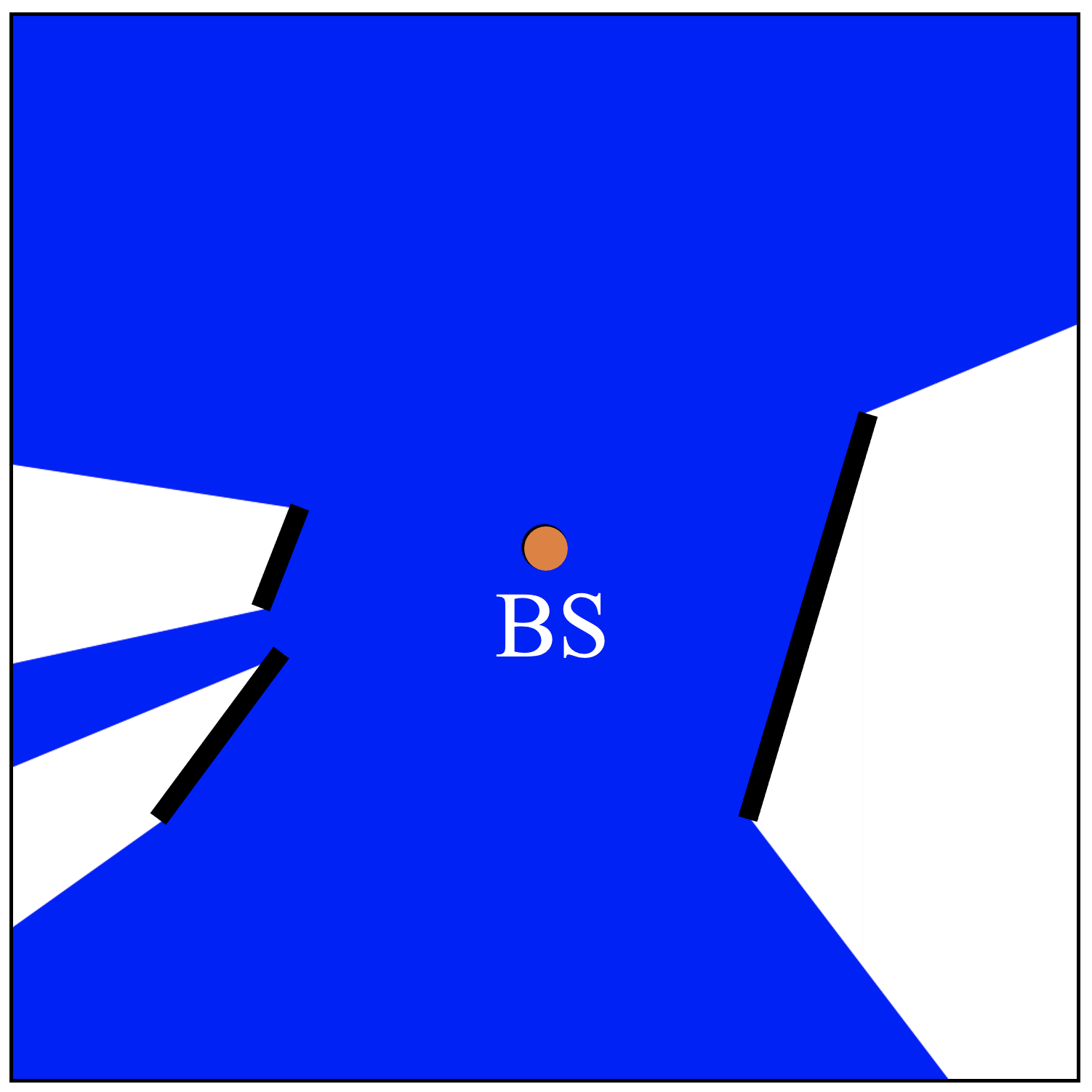}
    \label{coverage_ex3}
    }\quad
    \subfigure[$2$ RISs (wall-type).]{
    \includegraphics[width=0.2\columnwidth]{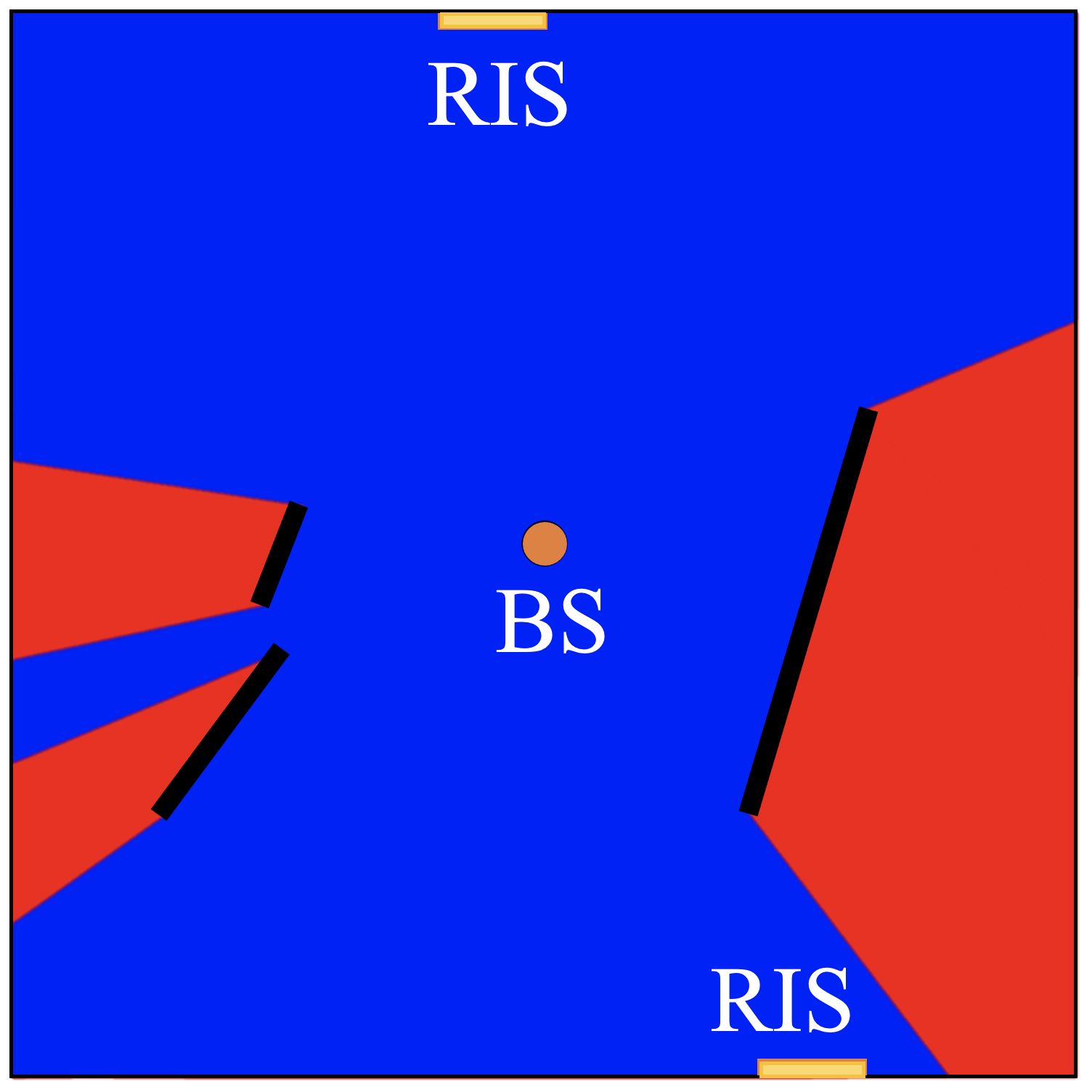}
    \label{coverage_ex4}
    }
    \caption{Comparison of the coverage regions with and without RISs.}
    \label{coverage_ex}
\end{figure*}

Since the coverage area depends on the area of $\mathcal{S}\setminus \mathcal{O}$, we introduce the normalized coverage area as $\operatorname{Area}\big(\mathcal{C}\left(\{\mathbf{q}_j\}_{j\in[1:J]}\right)\big)/\operatorname{Area}(\mathcal{S}\setminus \mathcal{O})$, where its maximum value is limited by one.
Then we present the normalized coverage area averaged over large enough indoor environment samples to capture the average performance in Figs. \ref{coverage_circletype} and \ref{coverage_walltype}.
Specifically, Fig. \ref{coverage_circletype} plots the average normalized coverage area for circular obstacles with respect to the number of RISs.
As seen in the figures, the coverage area achievable by `Full search in $\mathbf{Q}^*$' is very close to that of `Optimal', demonstrating that the proposed construction method of $\mathbf{Q}^*$ can efficiently reduce the number of RIS candidate positions while preserving the achievable coverage close to its optimal value. 
Also compared with the `Random' case, by carefully optimizing RIS positions, we can improve the coverage area by $20$ to $30$ percent for harsh indoor environment with the same number of RISs, see Fig. \ref{circletype:second}. Fig. \ref{coverage_walltype} plots the average normalized coverage area for wall-type obstacles with respect to the number of RISs and similar tendencies can be observed.

\begin{figure}%
\centering
\subfigure[$3$ circular obstacles ($I=3$).]{%
\label{circletype:first}%
\includegraphics[height=2in]{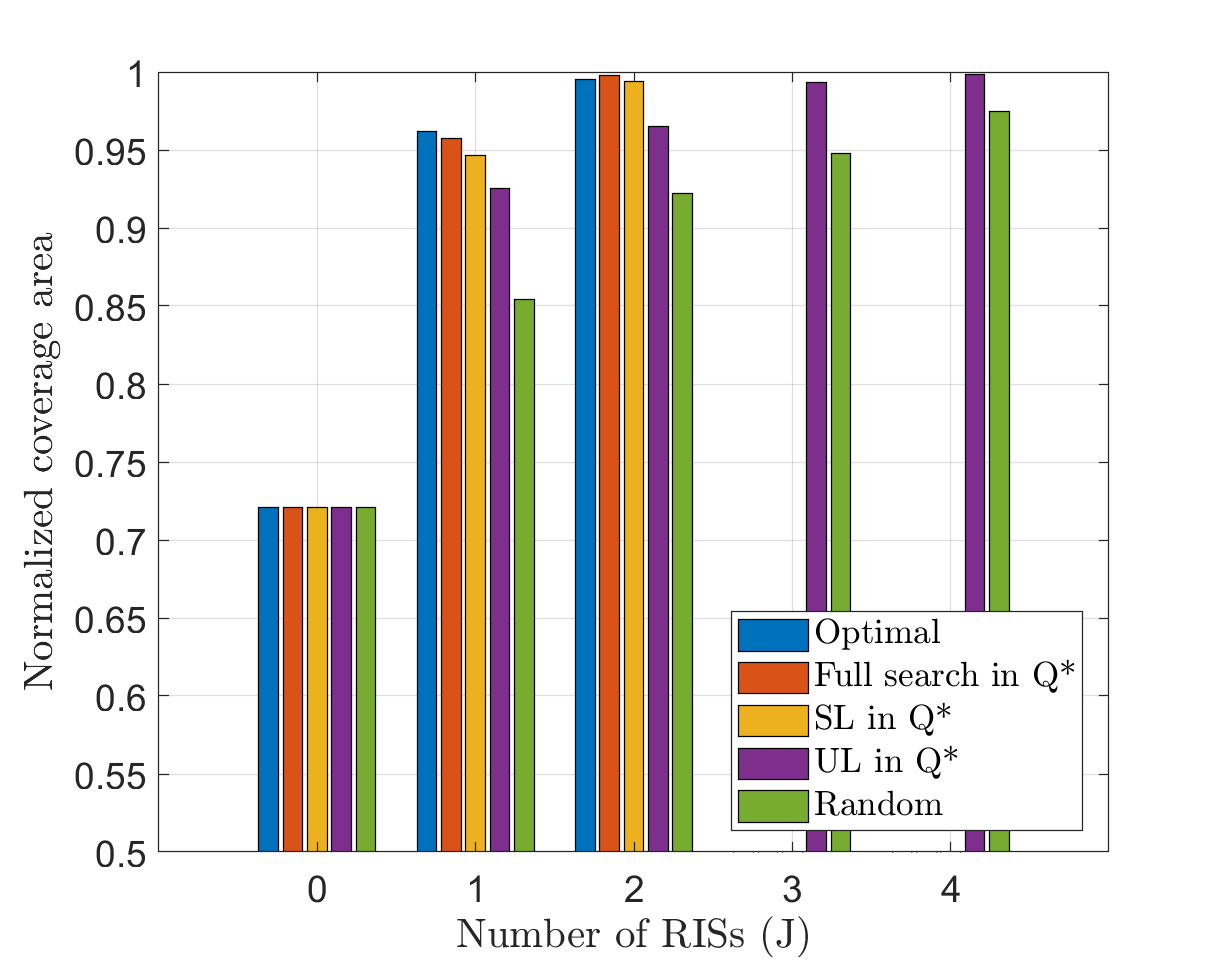}}%
\qquad
\subfigure[$5$ circular obstacles ($I=5$).]{%
\label{circletype:second}%
\includegraphics[height=2in]{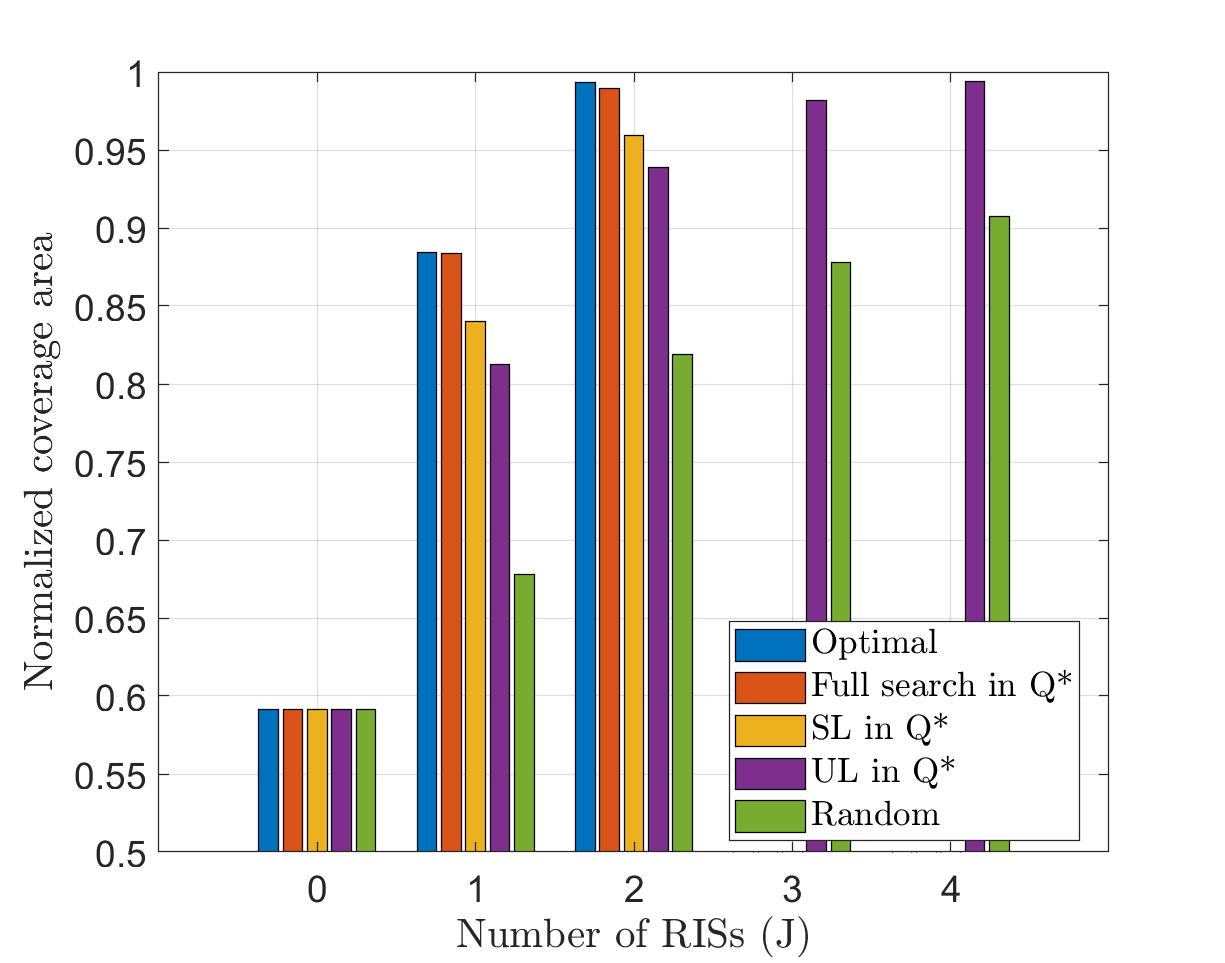}}%
\caption{Average normalized coverage with respect to the number of RISs for circular obstacles.}
    \label{coverage_circletype}
\end{figure}

\begin{figure}%
\centering
\subfigure[$3$ wall-type obstacles ($I=3$).]{%
\label{walltype:first}%
\includegraphics[height=2in]{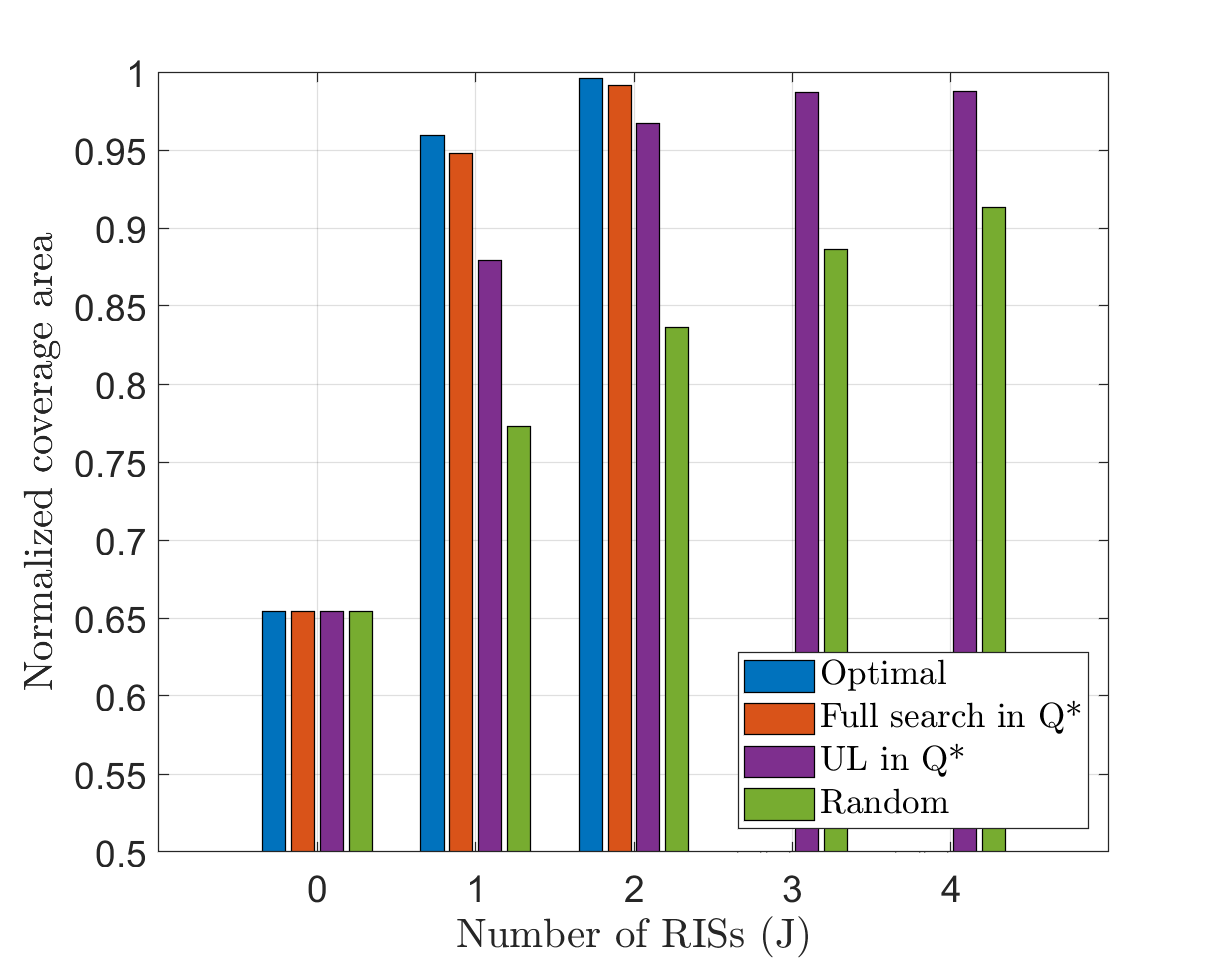}}%
\qquad
\subfigure[$5$ wall-type obstacles ($I=5$).]{%
\label{walltype:second}%
\includegraphics[height=2in]{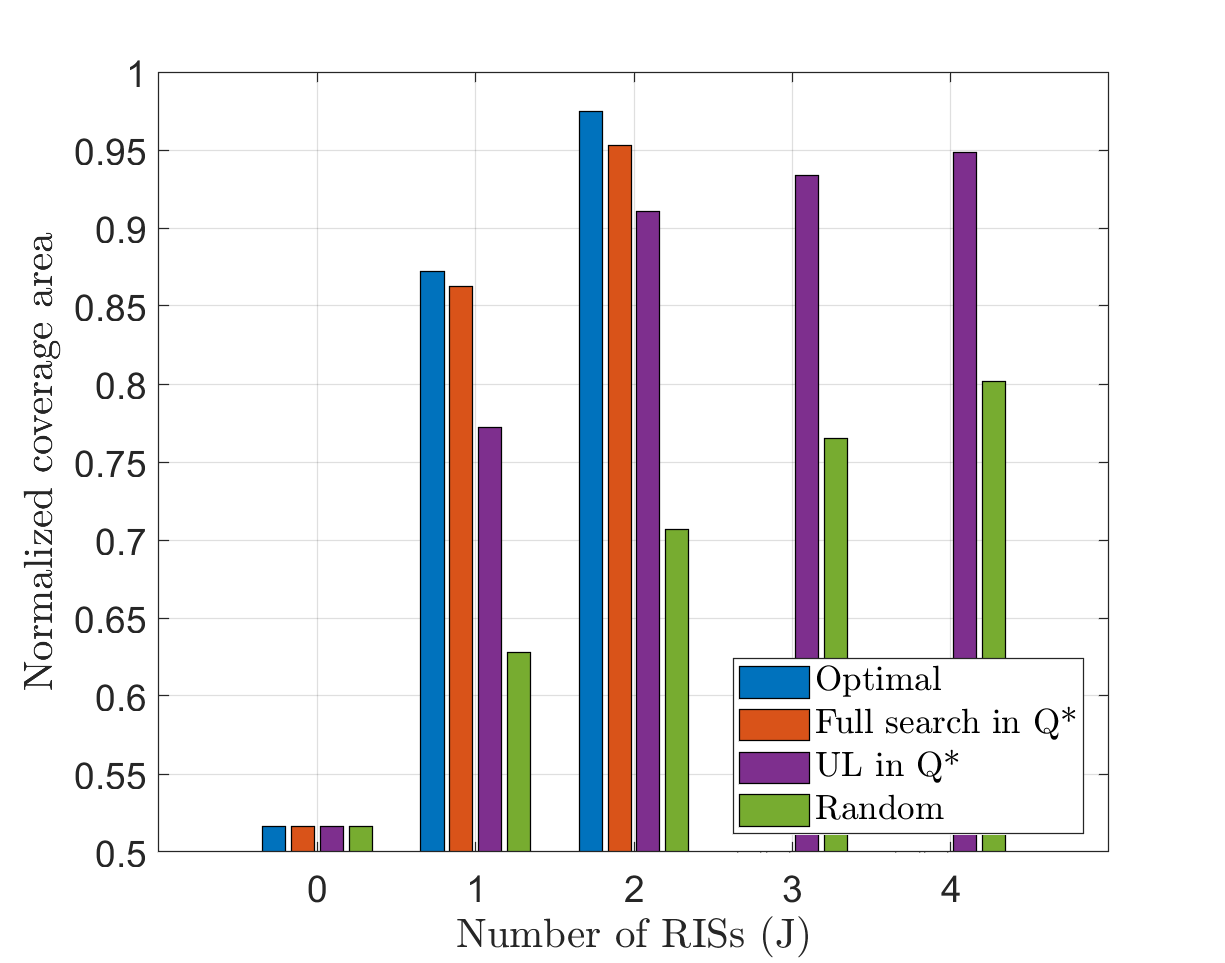}}%
\caption{Average normalized coverage with respect to the number of RISs for wall-types obstacles.}
    \label{coverage_walltype}
\end{figure}

\subsection{Evaluation of Sum Rate} \label{subsec:sum_rate}
In this subsection, we evaluate the achievable sum rate of the proposed RIS-aided hybrid beamforming scheme in Section \ref{sec:scheme_beamtraining}. 
The main simulation parameters for 3D indoor network environment are summarized in Table \ref{tab:environment_3d}.

\begin{table}\caption{3D indoor network parameters.}\label{tab:environment_3d}
\begin{center}
\scalebox{0.85}{
\begin{tabular}{c|c}
\hline
\multicolumn{1}{c|}{\bf Parameters} & \multicolumn{1}{c}{\bf Values} \\
\hline
\hline
 {Network region {($\mathcal{S}$)}} & {${[0,10{~}\text{m}]}\times{[0,10{~}\text{m}]}\times{[0,3{~}\text{m}]}$} \\
\hline
{BS position ($\mathbf{q}_0)$} & {$[5,5,3]$} \\
\hline
{BS orientation ($\alpha_0,\beta_0$)} & {xy-plane} \\
\hline
{Number of users $(K)$} & {$2$} \\
\hline
{Path loss exponent ($\gamma$)} & {$2$}  \\
\hline
{Channel model} &  {Saleh--Valenzuela model} \\
\hline
{Number of multi-paths } & {$3$} \\
\hline
{Variance of NLoS components ($\sigma_L^2$)} & {$-30$ dB} \\
\hline
{BS antenna $(L_1,L_2,M_1,M_2)$} & {$(2,1,4,4)$}\\
\hline
{Elements for each RIS $(N_1,N_2)$} & {$(8,8)$}\\
\hline
{RIS orientation ($\alpha_j,\beta_j$)} & {xz- or yz- plane}\\
\hline
\end{tabular}}
\end{center}
\end{table}

\subsubsection{Benchmark schemes}
For a performance upper bound, we consider the case of coherent analog beamforming at the BS and each RIS denoted by `Upper bound' in this subsection, which requires the exact CSI of BS--user channels and RIS--user channels. By comparing with `Upper bound', the effectiveness of the proposed beam scanning at the BS and each RIS can be verified.
For a performance lower bound, we consider the case that reflection coefficients are set randomly for all RISs denoted by `RND coefficient' in this subsection. For this case, the BS only performs beam scanning while RISs randomly set their reflection coefficients. By comparing with `RND coefficient', the effectiveness of the proposed RIS assignment and the corresponding construction of reflection coefficients can be verified.
Lastly, we also consider the case where there is no RIS denoted by `No RIS'.

In order to demonstrate the effect of the RIS placement on sum rates, we consider three RIS placement strategies: `OPT', `UL', and `RND' in this subsection.
Here, `OPT', `UL', and `RND' corresponds to the cases of `Optimal', `UL in $\mathcal{Q}^*$', and `Random' in Section \ref{subsec:coverage}, respectively.
More specifically, we firstly convert the 3D network into the corresponding 2D network and optimize the x- and y-axis position of RISs and then set the z-axis position as a predetermined value ($1.5$ m in simulation). 

\subsubsection{Comparison of sum rates}
Fig. \ref{fig:sum_rates_circular} plots the average sum rates with respect to SNR for circular obstacles, where $J=2$, $I=5$, $V_1=V_2=M$, and $W_1=W_2=N$ are assumed.
The same obstacle distributions in Section \ref{subsec:coverage} are used in simulation.
As seen in Fig. \ref{circletype:first1}, the average sum rate achievable by the proposed BS and RIS beam scanning is very close to its upper bound of coherent beamforming. Furthermore, by comparing `Proposed, OPT', `Proposed, UL', and `Proposed, RND', it can be seen that improved coverage areas provide improved sum rates. Also, the sum rate of `Proposed, UL' is very close to that of `Proposed, OPT' showing that the RIS placement optimization based on the proposed unsupervised learning with $\mathcal{Q}^*$ works very well for the sum rate maximization. From Fig. \ref{circletype:second2}, it can be seen that the proposed RIS assignment and then the construction of RIS reflection matrices to the direction of the served users can significantly improve sum rates compared to the case where the BS only performs hybrid beamforming combined with randomly selected RIS coefficients.   
Fig. \ref{fig:sum_rates_wall} plots the average sum rates with respect to SNR for wall-type obstacles and similar tendencies can be observed.

\begin{figure}%
\centering
\subfigure[Comparison with `Upper bound'.]{%
\label{circletype:first1}%
\includegraphics[height=2.2in]{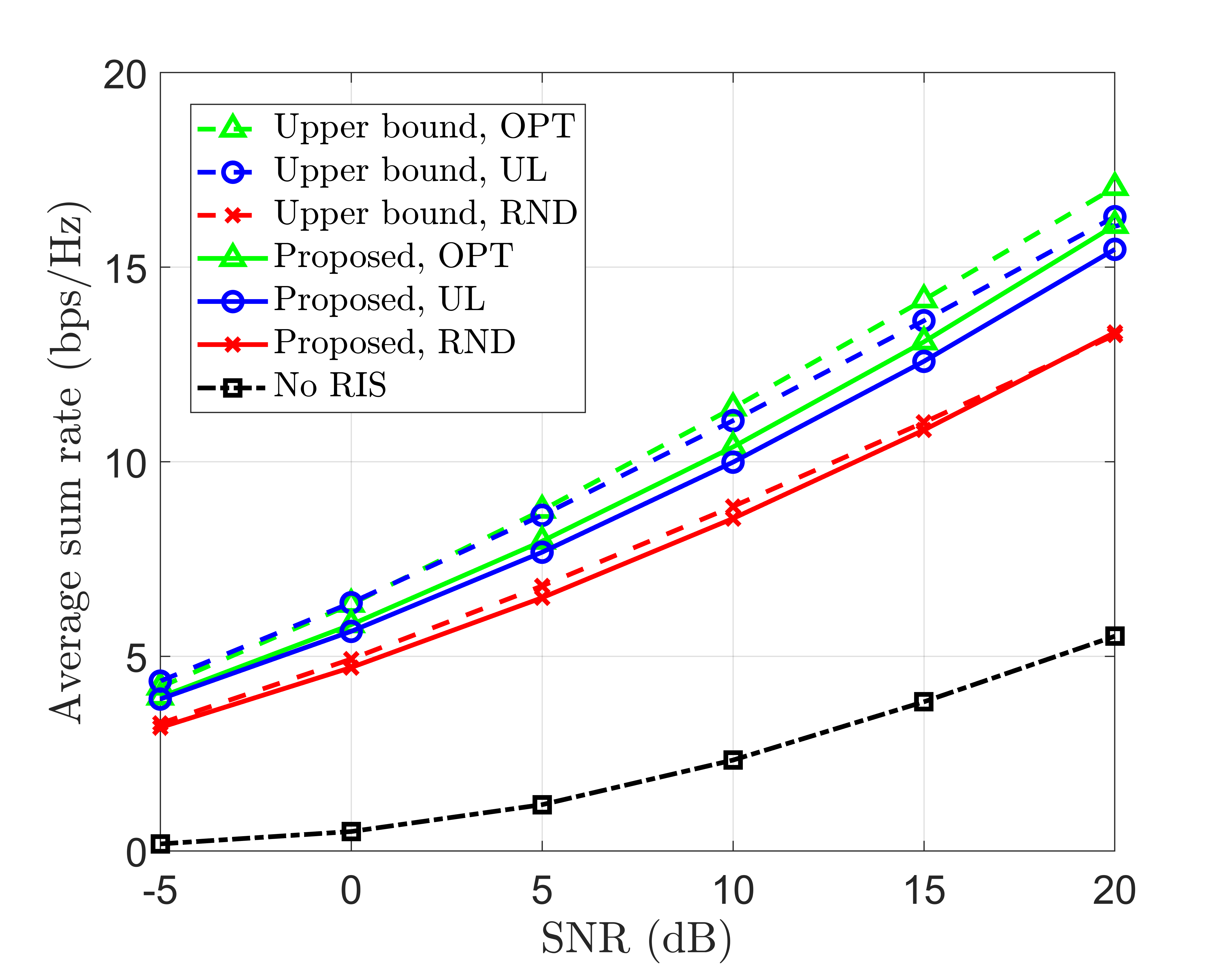}}%
\qquad
\subfigure[Comparison with `RND coefficient'.]{%
\label{circletype:second2}%
\includegraphics[height=2.2in]{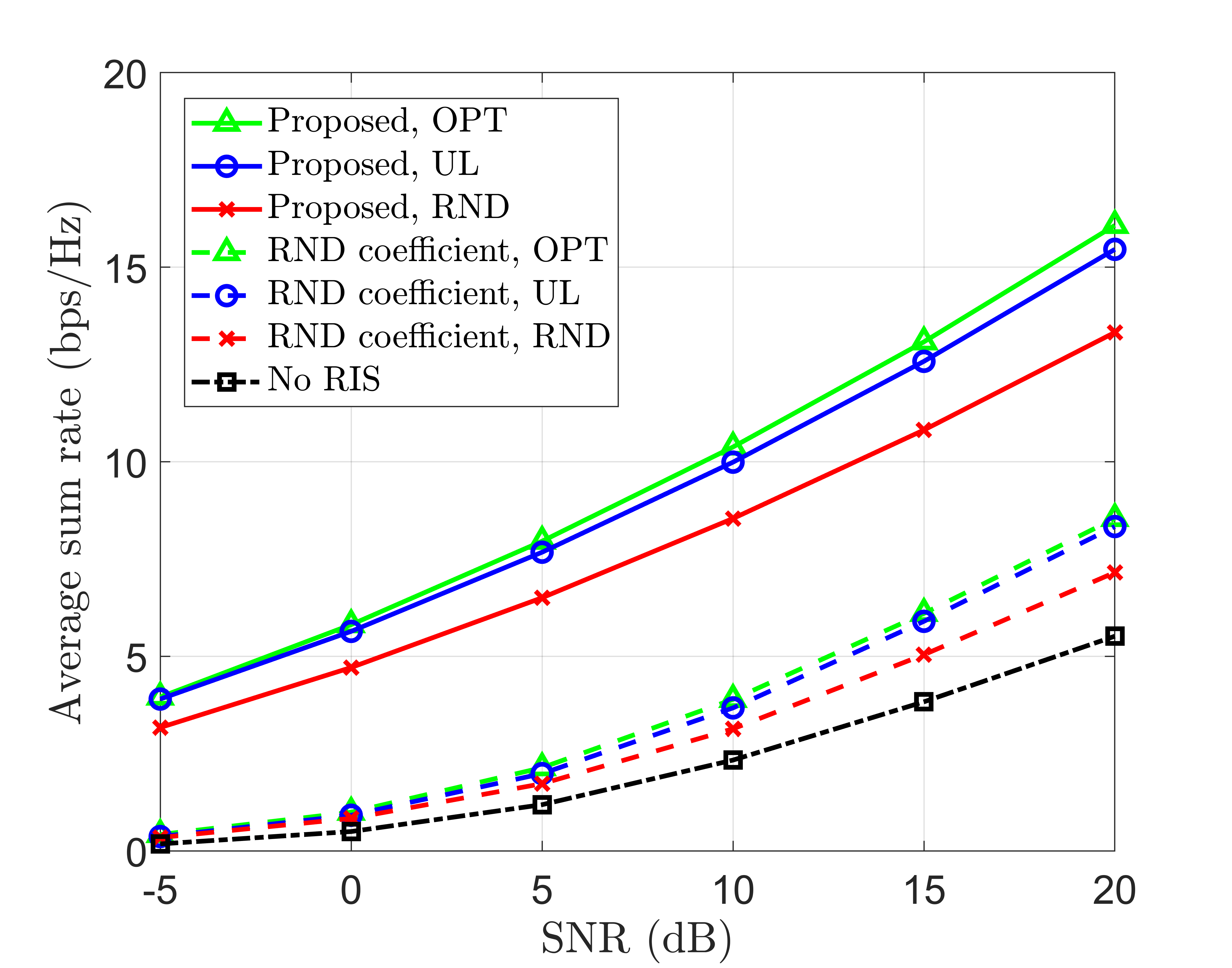}}%
\caption{Average sum rates with respect to SNR for circular obstacles.}
\label{fig:sum_rates_circular}
\end{figure}

\begin{figure}%
\centering
\subfigure[Comparison with `Upper bound'.]{%
\label{walltype:first1}%
\includegraphics[height=2.2in]{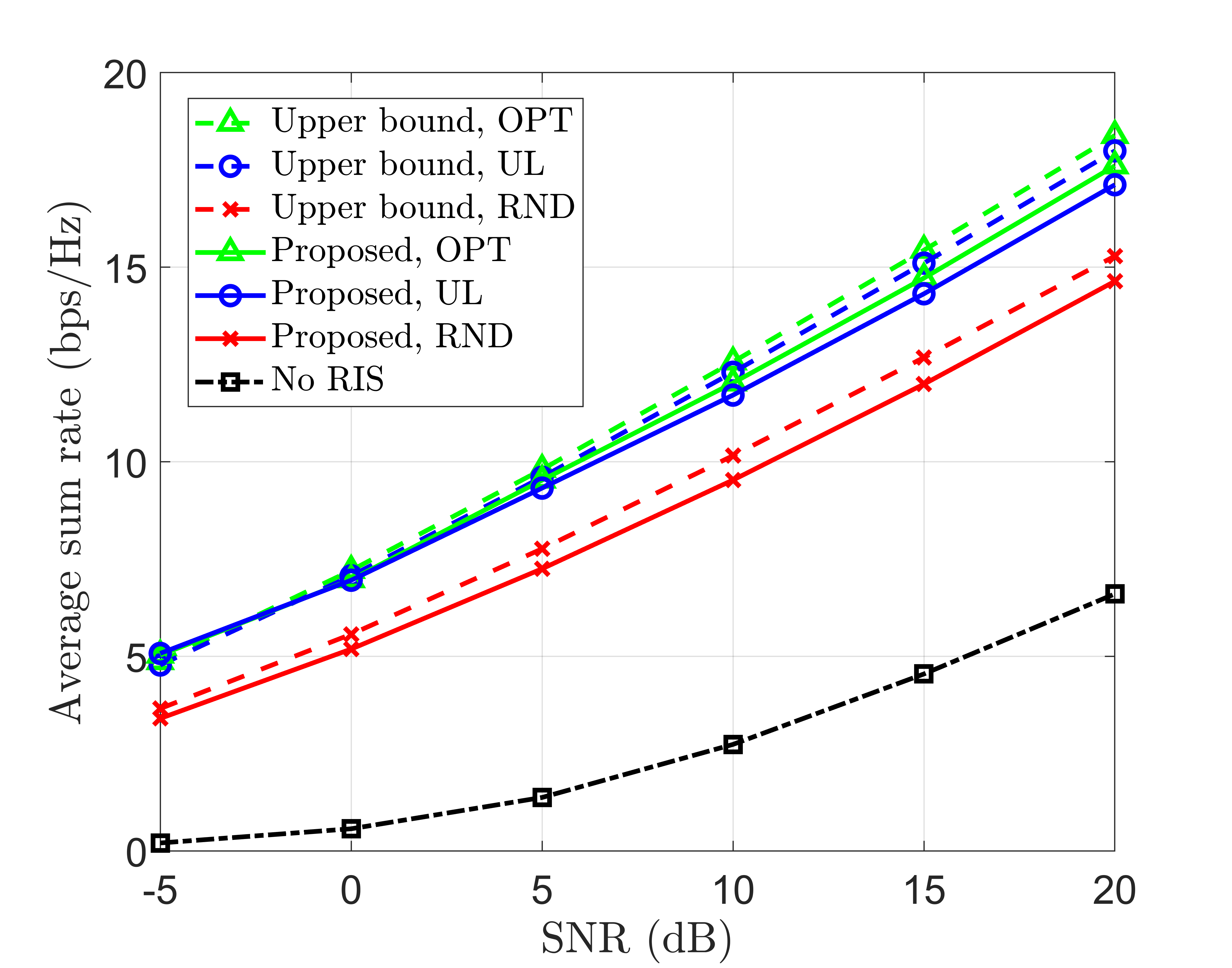}}%
\qquad
\subfigure[Comparison with `RND coefficient'.]{%
\label{walltype:second2}%
\includegraphics[height=2.2in]{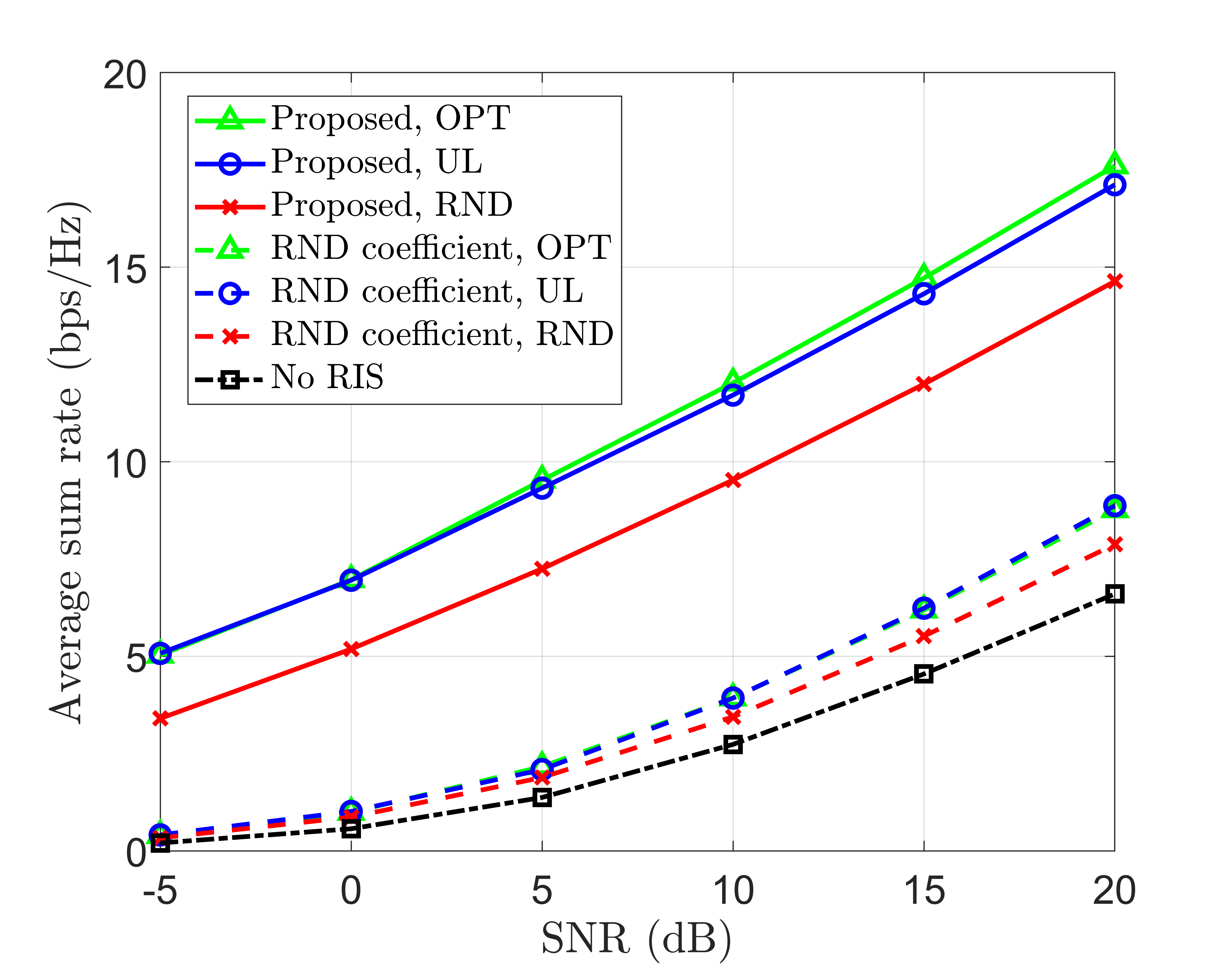}}%
\caption{Average sum rates with respect to SNR for wall-type obstacles.}
\label{fig:sum_rates_wall}
\end{figure}

\section{Conclusion} \label{sec:conclusion}
In this paper, we studied RIS-aided hybrid beamforming systems for mmWave or THz indoor communication.
Unlike most previous works, we considered the joint optimization of both RIS placement and the hybrid beamforming construction including the RIS reflection procedure.
Numerical results demonstrated that optimally deployed RISs can efficiently improve both the coverage areas and sum rates, which are crucially important for indoor ultra-high frequency communication. 

An interesting direction for future work would be to extend the current framework to the case of multiple BSs and then jointly optimize the placement of both BSs and RISs. A key technique issue will be how to construct an optimal or near-optimal placement method with low computational complexity, scalable with an increasing number of BSs and/or RISs.  

\bibliographystyle{IEEEtran}
\bibliography{IEEEabrv,References}

\end{document}